\documentclass[sigconf]{acmart}

\acmConference[]{}{}

\usepackage[utf8]{inputenc}
\usepackage{url}

\usepackage{microtype}
\usepackage{caption}
\usepackage{subcaption}
\usepackage{booktabs} 
\usepackage{amsmath}
\usepackage{amsthm}
\usepackage{nicefrac}
\usepackage{listings}
\usepackage{balance}

\usepackage{enumerate}

\usepackage{hyperref}       
\usepackage{url}            
\usepackage{booktabs}       
\usepackage{xcolor}         
\usepackage{adjustbox}

\usepackage{header}
\usepackage{math_commands}

\usepackage{wrapfig}

\usepackage[textwidth=18mm]{todonotes}

\newif\ifappendix
\appendixfalse

\title{Ranking with Popularity Bias: User Welfare under Self-Amplification Dynamics}

\author{Guy Tennenholtz }
\affiliation{%
  \institution{Google Research}
}

\author{Nadav Merlis}
\affiliation{%
  \institution{CREST, ENSAE}
}

\author{Martin Mladenov}
\affiliation{%
  \institution{Google Research}
}

\author{Robert L. Axtell}
\affiliation{%
  \institution{George Mason University}
}

\author{Craig Boutilier}
\affiliation{%
  \institution{Google Research}
}

\begin{document}


\begin{abstract}
   While \emph{popularity bias} is recognized to play a crucial role in recommmender (and other ranking-based) systems, detailed analysis of its impact on collective \emph{user welfare} has largely been lacking. We propose and theoretically analyze a general mechanism, rooted in many of the models proposed in the literature, by which item popularity, item quality, and position bias jointly impact user choice. We focus on a standard setting in which user utility is largely driven by item quality, and a recommender attempts to estimate it given user behavior. Formulating the problem as a non-stationary contextual bandit, we study the ability of a recommender policy to maximize user welfare under this model. We highlight the importance of exploration, \emph{not to eliminate} popularity bias, but to mitigate its negative impact on welfare. We first show that naive \emph{popularity-biased recommenders} induce linear regret by conflating item quality and popularity. More generally, we show that, even in linear settings, \emph{identifiability} of item quality may not be possible due to the confounding effects of popularity bias. However, under sufficient variability assumptions, we develop an efficient optimistic algorithm and prove efficient regret guarantees w.r.t.\ user welfare.  We complement our analysis with several simulation studies, which demonstrate the negative impact of popularity bias on the performance of several natural recommender policies.
\end{abstract}

\maketitle

\section{Introduction}

The study of growth dynamics in multi-agent systems has a long history that spans many disciplines and finds application in various domains. Among these are dynamic models of the wealth of individuals in a community \citep{pareto1896wealth}, the size of firms in a labor economy \citep{axtell2006firmsize}, and the popularity of individuals or items in social/citation networks \citep{barabasi1999emergence}. A prominent finding is that growth is often driven by positive feedback mechanisms, i.e., the larger, more popular, or richer an entity is, the more likely it is to grow further. Such ``rich-get-richer'' phenomena have been explained by a variety of mechanisms, such as \emph{preferential attachment} \citep{simon1955powerlaw, barabasi1999emergence}, most of which generate \emph{power-law} size distributions.

In \emph{recommender ecosystems}, this phenomenon has often been described by the umbrella term \emph{popularity bias} \citep{bellogin2017statistical,abdollahpouri2019managing,abdollahpouri2021user,wei2021model,elahi2021investigating,ahanger2022popularity}. Such bias is often assumed to arise when a \emph{recommender system~(RS)} recommends popular items more frequently, which in turn increases their consumption by users, further amplifying their popularity \citep{abdollahpouri2019managing,abdollahpouri2020multi}. This can induce or magnify common long-tail effects.

While various mechanisms to explain popularity bias have been proposed \citep{abdollahpouri2019managing,zheng2020disentangling,chen2022co,he2022mitigating}, few have directly incorporated or theoretically analyzed the interventions needed to mitigate the impact of the \emph{RS's ranking policy}. Indeed, the growth dynamics of popularity in recommender ecosystems are complicated by the presence of the RS's ranking algorithm \citep{baeza2005query,zoghi2017online,wang2018streaming,amato2019sos}. RS policies (or \emph{rankers}) can both amplify popularity feedback (e.g., the phenomenon of ``going viral'' \citep{west2011going}) as well as mitigate it (e.g., using mechanisms involving fairness of exposure \citep{singh2018fairness}). We argue that the ranker must be incorporated directly into any mechanistic account of popularity~bias. 

In this work we propose a theoretical framework for a more nuanced study of popularity bias in RSs, with an emphasis on its impact on collective user utility, or \emph{social welfare}, and its interaction with various types of RS policies. Much of our framework is standard, assuming that a user's utility is primarily dictated by the \emph{quality} of the recommended items they select, and that the RS must attempt to assess this quality via its observations of user selections.\footnote{We use the term ``quality'' to reflect inherent quality, relevance, and any other RS- and population-independent factors that determine a user's satisfaction with a recommended item in a given context.} 
We outline a general mechanism by which an item's popularity can \emph{influence} a user's choice/consumption behavior in conjunction with other relevant factors; namely, item quality and position bias. 

Our work emphasizes the potential \emph{confounding} of popularity and quality, inducing an identifiability problem for any social-welfare optimal RS policy. We address this by formalizing our problem as a \emph{nonstationary contextual bandit} \citep{faury2020improved,abeille2021instance,amani2021ucb}. We provide regret analyses to show whether, and under what conditions, a ranking policy can optimize long-term expected user utility. Specifically, we show that \emph{an optimal ranking policy is not achievable} in the general case due to the confounding effect of popularity on the estimation of item quality. However, we prove that, under specific variability assumptions, popularity bias is \emph{statistically identifiable}. This allows us to develop an \emph{efficient, optimistic} algorithm with sublinear regret which, under these conditions, decouples quality from popularity, while exploiting the positive aspects of the feedback mechanism. We also describe and analyze other natural RS ranking policies which can either \emph{amplify or mitigate the negative welfare effects of popularity bias} under various observability assumptions. In particular, we show that \emph{exploration} is essential if an policy is to overcome such effects.

Our contributions are as follows. In \Cref{section: problem setup} we propose a novel framework that incorporates item quality, item popularity, and rank/position bias in a non-stationary contextual user choice model. In \Cref{section: overcoming popbias} we study the optimization and exploration problem under this model, and provide the \emph{first theoretical regret guarantees for popularity bias}. In particular, we show that the ranking problem is generally not identifiable (even under a linear class assumption) in the presence of popularity-biased users, and prove lower bounds. We propose and analyze an asymptotically optimal ranking policy, and show that, under a sufficient diversity assumption, an efficient solution can be developed. Indeed, we formulate a UCB-style algorithm and prove efficient regret guarantees.
In \Cref{sec:experiments}, we describe simulation studies to assess the performance and dynamics of popularity bias under our methods. We also compare several natural baselines, showing that they are prone to converge to suboptimal item-selection distributions with linear regret. 
Our research represents a significant step in the understanding of the complex dynamics of self-amplification of popularity in RSs and its implications on user welfare, and provides the first regret guarantees for ranking policies in settings with popularity bias.

\section{Popularity Feedback Loops}
\label{sec:background}

Before addressing the problem of optimal rankers, we first describe informally how popularity bias can negatively impact the expected collective utility an RS generates for its users, i.e., social welfare. To illustrate the factors that contribute to popularity bias, we consider three distinct ways in which it can be induced: (i) by an RS ranking policy, (ii) via inherent user bias, or (iii) by the behavior of suppliers of items (e.g., product vendors, content creators). Our model in Sec.~\ref{section: problem setup} incorporates both (i) and (ii)---we include mechanism (iii) here only for completeness. We formalize and assess the different forms of ranking policies discussed here in \Cref{sec:experiments}.
%

\subsection{Popularity Bias Caused by Rankers}
\label{section: pop-biased rankers}

Contemporary RSs often attempt to maximize user utility in a myopic or greedy fashion, causing them to underexplore. Moreover, they seldom account for certain structural biases, such as popularity bias, though rank/position bias is often included in some analyses. Such misspecifications can generate undesirable outcomes and induce linear regret 
(see below).

For simplicity, suppose users are unbiased by popularity, only selecting items based on their quality and rank position. In this scenario, popularity bias might be induced by a \emph{popularity-driven ranker}, which ranks items presented to a user according to their current popularity (e.g., the number of selections that item has so far). This ranking strategy is natural for an RS which cannot observe item quality, but attempts to indirectly estimate item quality by using \emph{popularity as a proxy}.

Such popularity-driven rankers gradually \emph{amplify} the positions of higher-ranked items. Particularly, though users' selections are positively correlated with quality, the induced feedback loop may lead to higher ranked items to increase in popularity due to the rank-bias in user selections. This feedback loop can converge to a state where lower-quality items are promoted at the expense of higher-quality items, and is among the most widely discussed popularity-bias-inducing mechanisms in the literature \citep{salganik2006experimental,ciampaglia2018algorithmic,yalcin2021investigating}. We demonstrate the negative effect this can have on social welfare in Sec.~\ref{sec:experiments}.



\subsection{Popularity-Biased Users}
\label{section: pop-based users}

We next turn to situations where feedback loops emerge due to \emph{popularity bias in user selections}.
In this setting, a low-quality, but more popular, item might have higher probability of selection by a user than a higher-quality, less popular item. We call such users \emph{popularity-biased}, but distinguish two reasons for this bias:
\begin{itemize}
    \item First, a user may derive utility based only on the quality of the item selected, but may be unable to directly or accurately observe an item's quality (e.g., a video's quality may only be assessed fully \emph{after} watching it). In this case, popularity may be used as a noisy indication of quality by the user.\footnote{This is more likely to be true in RSs that provide an indication of popularity, e.g., view counts, alongside a recommended item.}
    \item Second, a user may actually derive some utility from an item's popularity. This might be the case due to ``network effects'' when, say, the value of consumed content (e.g., music, movies, news) is greater when it can be discussed with more friends, or when it increases connection to an in-group.
\end{itemize}
In either case, users may select items based on popularity at the expense of quality, or other utility-bearing attributes.

Since most RSs exploit user selections to estimate item quality/user utility, any ranker that ignores the effect of popularity, will induce some form of popularity bias. Indeed, we will examine such a \emph{popularity-oblivious ranker} in Sec.~\ref{sec:experiments}. Not surprisingly, when user utility depends only on item quality, but selections are biased by popularity (the first condition above), the induced popularity bias does significant harm to social welfare. That said, when user utility is largely dictated by item popularity (the second condition), the popularity-oblivious ranker may perform reasonably well. However, item quality will play a significant role in any realistic utility model and the inability to disambiguate popularity and quality will serve as an impediment to the generation of high social welfare.

\subsection{Resource (or Skill) Bias}

For completeness, we consider a final mechanism in which popularity bias is implicit in \emph{dynamic item quality}. Consider a model in which users consume new, changing items (e.g., dynamic content) offered by a set of \emph{providers}.
In many cases, the popularity of the items offered by a provider influences their \emph{proficiency} at item generation. This may be due to, for instance, greater resource availability (e.g., due to revenue generation) or improved generation skill (e.g., due to more opportunities, social incentives, etc.). 

Such \emph{resource (or skill) bias} can be modeled by the \emph{increase in quality}
of a provider's new items
w.r.t. the number of selections of their past items.
A typical RS ranker, which attempts to maximize user utility greedily, recommends items with highest estimated quality. Nevertheless, the increases in quality offered by any item provider are driven directly by the probabilistic realization of past selections of their items. This in turn may induce a feedback loop which amplifies the increase in quality of the items offered by providers with more past selections.
Moreover, ``small'' or less popular providers are not given the same opportunity to improve their item quality, as they do not yet have the same level of resources or skill to do so. 

Of note is the fact that this type of popularity bias may have both positive and negative effects. For example, this bias may actually improve long-term user welfare when skill/quality increases non-linearly with popularity. On the other hand, it may have a negative impact if providers stop creating items due to a small audience of users, or other negative effects on item diversity or ``tail" items. Moreover, issues of provider equity and fairness will often be at play \citep{abdollahpouri2020connection}. While we do not dwell further on this, we note that such resource bias can be directly modeled in the user choice model  we propose below.

\section{Problem Definition}
\label{section: problem setup}

We begin by providing our problem formulation, including the introduction of a general \emph{user choice model} which incorporates item quality (or affinity), rank/position bias, and importantly, \emph{popularity bias} in the determination of a user's probability of selecting or consuming a specific recommended item. We later show how popularity bias complicates the computation of optimal RS policies as it confounds the ability of the RS to estimate inherent item quality. 

\subsection{Setup}
We use lower case bold letters for vectors (e.g., $\bs{x}$) and upper case bold letters for matrices (e.g.,  $\bs{V}$). For $\bs{u} \in \R^M$ let
\begin{align*}
    z_i(\bs{u})
    =
    \frac{\exp\brk[c]*{u_i}}{1 + \sum_{j=1}^M \exp\brk[c]*{u_j}} \, , \quad 1 \leq i \leq M,
\end{align*}
and $z_0(\bs{u}) = 1 - \sum_{i=1}^M z_i(\bs{u})$.
Finally, let $\bs{z}(\bs{u}) = \brk*{z_1(\bs{u}), \hdots, z_M(\bs{u})}^T$ be the softmax function.

Let $\gD$ be a corpus of \emph{recommendable} items (e.g., articles, music tracks, videos, products) and $\gU$ be a set of users. At each time~$t$, a user $u_t \in \gU$, sampled from some distribution $\gP_{\gU}$, seeks a recommendation from the RS. The RS uses a \emph{ranking policy (or ranker)} $\gR$ to select $M$ items from $\gD$, ranks them in some chosen order, and presents the resulting (ordered) \emph{slate} of items $\bs{s}_t$ to the user $u_t$. The user then selects at most one item from the slate. We write $\bs{s}_t = \brk*{I_{1,t}, I_{2,t}, \hdots, I_{M,t}}$, where $I_{i,t}$ is the item at position/rank $i \leq M$ at time $t$. Let $\gS$ be the set of possible slates. 

Let $c_t \in \brk[c]*{I_{i,t}}_{i=1}^M \cup \brk[c]*{\emptyset}$ denote the item selected at time $t$, and $\bs{c}_{1:t} = \brk*{c_1, \hdots, c_{t}}$ the sequence of selections up to time $t$. We use $\mathbb{E}_{t}$ to denote expectation w.r.t.\ all histories through time $t$. Also let
\begin{align*}
n_t(I) = \sum_{k=1}^{t-1} \mathbbm{1}\brk[c]*{\text{item $I$ was selected at time $k$}}
\end{align*}
be the number of selections of item $I$ up to time $t$ (exclusive). Finally, let
$\bs{n}_t(\bs{s}) = \brk*{n_t(I_{1}), \hdots, n_t(I_{M})}^T$, and $\bs{n}_t = \bs{n}_t(\bs{s}_t)$.

The goal of the ranking policy is to choose a slate of items which induces a user selection that has high value to the user. This is complicated by two factors. First, the RS is not given the value of the items to any specific user---it must estimates these values based on past user selections. In other words, it must solve the underlying \emph{exploration (or bandit)} problem. Second, the process by which users make selections from slates may itself be rather involved, complicating this exploration problem. We next outline the selection (i.e., \emph{user choice}) model we adopt, then turn to the RS objective and a formalization of the induced exploration problem.

\subsection{User Choice Model} 
A \emph{user choice model} determines the probability with which a user $u$ selects a specific item $I$ from a slate $\bs{s}$. We introduce a choice model which incorporates quality, position, and popularity bias as three factors influencing a user's choice. We note that each of these factors is separately discussed in the literature \citep{tejeda2014quality,abdollahpouri2017controlling,collins2018position}. However, as we show later, when user choice is influenced by all three (and specifically, quality and popularity), the underlying exploration problem may be not identifiable.

The first factor, \emph{quality}, reflects the inherent 
value a user derives from an item,
and is captured by a \emph{quality bias} function ${\qualitybias: \gU \times \gS \mapsto \R^M}$ that, for any user $u$, specifies the user-specific quality of each item in slate $\bs{s}$. Quality bias is perhaps the most common factor in classic models of user choice \citep{mcfadden1973conditional,mcfadden1981econometric}, RS research \citep{balabanovic1997fab,li2005hybrid,mooney2000content}, and numerous other areas.

The second factor, \emph{rank bias}, is given by a vector $\rankbias \in \R^M$. This factor determines the impact of an item's position in a slate on user choice. Rank bias is another factor that is commonly modeled in user choice models in search \citep{wang2016learning,joachims2017unbiased,wang2018position}, computational advertising \citep{dave2014computational} and RS research \citep{vargas2011rank,guo2019pal}. For example, the cascade model \citep{kveton2015cascading,zong2016cascading,richardson2007predicting,craswell2008experimental,kveton2022value}, which has been widely used to model user choice in recommendation, models the inherent bias of users to select higher ranked positions in a slate of recommended items.

Finally, the user-dependent \emph{popularity bias} depends on the sequence of selected items (by all users) up to time $t$, and is given by $\popularitybias_t: \gU \times \gS \times \gD^{t-1} \mapsto \R^M$.  Popularity bias has been studied and acknowledged to be a major factor of critical concern for RSs \citep{abdollahpouri2020popularity,perez2019analysis,celma2008hits,jannach2015recommenders,niemann2013new}.

Given the three factors of quality, popularity and position bias, we are now ready to define our user choice model. For this, we define the \emph{user disposition} at time $t$ by
\begin{align*}
    \disp_t(u_t, \bs{s}_t, \bs{c}_{1:{t-1}}) = \qualitybias\brk*{u_t, \bs{s}_t} + \popularitybias_t\brk*{u_t, \bs{s}_t, \bs{c}_{1:{t-1}}} + \rankbias.
\end{align*}
Disposition combines the bias factors additively. While non-linear combinations of these biases are possible, we intentionally choose to analyze an additive model to lay the groundwork for theoretical research into popularity bias. This model allows us to develop rigorous theoretical insights, including the first upper/lower bounds and regret guarantees for ranking with popularity bias. 

User $u_t$ at time $t$ selects the item at position $i$ with probability
\begin{align}
\label{eq: click probability}
    z_i\brk*{\disp_t(u_t, \bs{s}_t, \bs{c}_{1:{t-1}})}, \quad
    1 \leq i \leq M;
\end{align}
and $z_0(\disp_t(u_t, \bs{s}_t, \bs{c}_{1:{t-1}})) = 1 - \sum_{i=1}^{M} z_i(\disp_t(u_t, \bs{s}_t, \bs{c}_{1:{t-1}}))$ is the probability of no selection.
Thus, selection probabilities are given by the softmax function, $\bs{z}\brk*{\disp_t(u_t, \bs{s}_t, \bs{c}_{1:{t-1}}))}$. 
\begin{remark*}
We note that previous work has provided empirical evidence for the applicability of our causal structure, conflating popularity and quality (see e.g., \citet{zheng2021disentangling}). Our work emphasizes the theoretical underpinnings of these factors, providing the first lower and upper bounds for this problem, including an impossibility result, as well as regret guarantees for an efficient ranking policy.
\end{remark*}


\subsection{Optimality Criterion.} 
\label{sec: optimality criteria}
An RS policy (or ranker) $\gR: \gU \times \gH_t \mapsto \gS$ determines the item slate shown to a user given an interaction history (here $\gH_t$ is the set of all length $t-1$ histories). Once presented, the user makes a selection using the choice model in \Cref{eq: click probability}. We measure the \emph{value} of $\gR$ using the total expected \emph{user utility} of selected items. Specifically, conditioned on $u_t\! =\! u, \popularitybias_t\! =\! \bs{b}, \bs{s}_t\! =\!\bs{s}$, we define
\begin{align*}
    v^{\gR}_t(u, \bs{b}, \bs{s})
    \!=\!\!\! \mathop{\mathop{\Exptn}_{u_{t+1}, \hdots, u_T}}_{\gR}
    \brk[s]*{
    \sum_{k=t}^T 
    \sum_{i=1}^M
    \! \utility\brk*{u_k, \bs{s}_k}_i
    \!\cdot\!
    z_i\brk*{\disp_k(u_k, \bs{s}_k, \bs{c}_{1:{k-1}})}
    } ,
\end{align*}
where $\utility: \gU \times \gS \mapsto \R$ is some measure of user utility. We generally equate user utility with item quality, i.e., $\utility \equiv \qualitybias$, though other measures are possible, e.g., $\utility \equiv \qualitybias + c \cdot \popularitybias$, for some $c > 0$.  Importantly, the RS only observes user selections, not utility (e.g., user-specific item quality), and typically must \emph{attempt to infer the latent utility} $\utility$ to optimize its policy. Finally, let $v^{\gR}_t = \expect*{u \sim P_{\gU}}{v_t^{\gR}(u, 0, S^{\gR}_1)}$ be the expected value of $\gR$ when popularity bias is initialized to zero. 

An \emph{optimal ranker} $\gR^*$ maximizes the expected return at time~$t$, defined recursively as ${
    \bs{s}_t^*(u,\bs{b}) \in \arg\max_{\bs{s} \in \gS} v_t^{\gR^*}(u, \bs{b}, \bs{s})}
$.
The \emph{cumulative regret} of a ranker $\gR$ at time $T$ is then defined by:
$$
    \text{Reg}_{\gR}(T)
    =
    v^{\gR^*}_1 - v^{\gR}_1.
$$
Our goal is to develop rankers with low regret, where possible. As we will see, this depends critically on the presence, interaction and RS observability of the bias terms that determine user disposition.

\begin{figure*}[t!]
\centering
\includegraphics[width=0.9\linewidth]{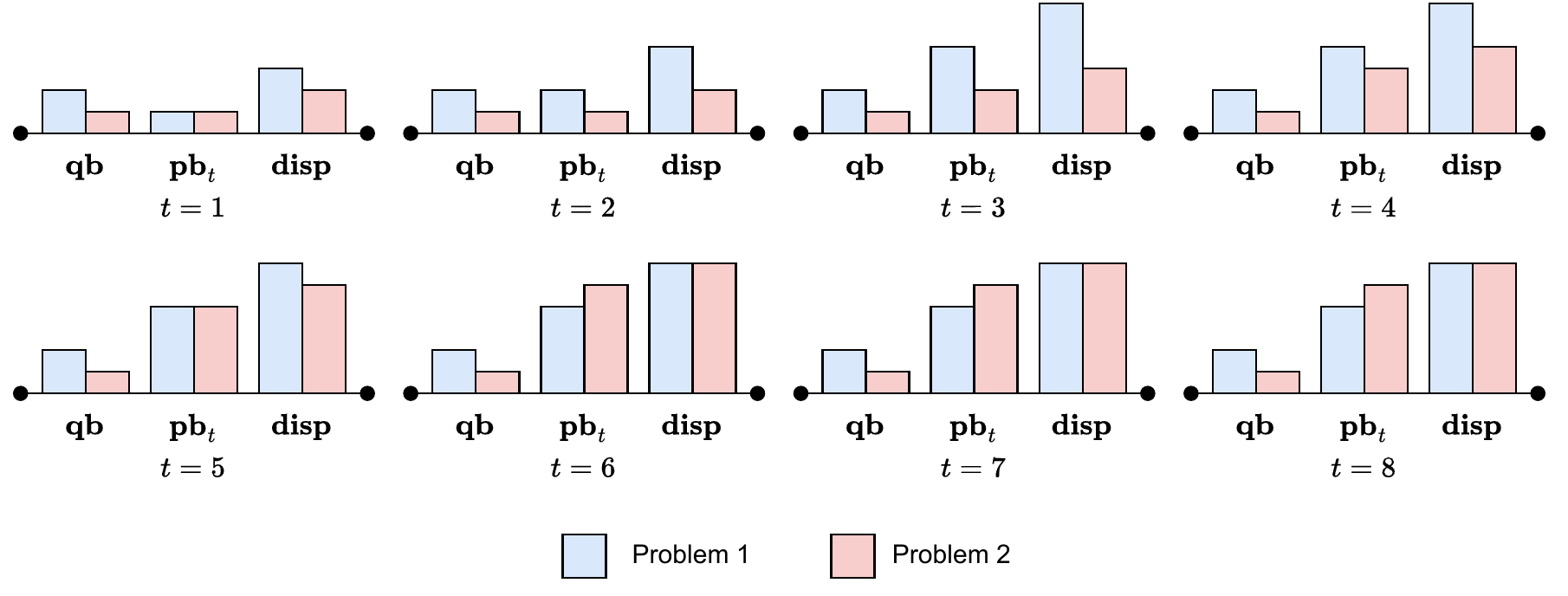}
\caption{\small A plot depicting the non-identifiability problem of ranking with popularity bias. As the popularity bias of Problem 1 (blue) and Problem 2 (red) increase, it eventually reaches saturation in both problems. In this example, the rank bias is zero, and the saturation of the two problems brings about the same user disposition $\disp = \qualitybias + \popularitybias$. After iteration $t=6$, the distribution of selection probability in the two problems is the same, rendering the quality bias not identifiable.}
\label{fig: impossibility}
\end{figure*}

\subsection{Additional Modeling Assumptions} 
We outline several key assumptions used in our analysis below. Many of our assumptions are standard in the bandit literature \citep{amani2021ucb}. First, we assume that all of the user choice factors are bounded.
\begin{assumption}[Bounded Disposition]
\label{assump: disp boundedness}
For any $u \in \gU$, $\bs{s} \in \gS$, $1 \leq i \leq M$, $\qualitybias\brk*{u, \bs{s}}_i, \popularitybias\brk*{u, \bs{s}, \bs{c}}_i, \rankbias_i \in [-1,1]$.
\end{assumption}
This assumption can be relaxed for bounds in arbitrary intervals, adding another factor to the regret related to scaling of these factors (see \citet{abeille2021instance}). We focus on the interval $[-1, 1]$ here for clarity of our results, though our approach is agnostic to this scaling.

Next, we assume that the popularity bias is monotonically increasing in the number of selections induced by the RS, and bounded. 
\begin{assumption}[Monotonic and Bounded Popularity]
\label{assump: popularity bias}
For any $1 \leq i \leq M$, $t \in \sN$, $u \in \gU$, $\bs{s} = \brk*{I_1, \hdots, I_M} \in \gS$, ${\bs{c} = \brk*{c_1, \hdots, c_{t-1}} \in \gD^{t-1}}$,
\begin{align*}
    \popularitybias_{t}\brk*{u, \bs{s}, \bs{c}}_i
    =
    \min\brk[c]*{
        \sum_{k=1}^{t-1}
        \alpha_{k,i}(I_{i}, c_k)
        ,
        ~~b_i(u, \bs{s})
        },
\end{align*} 

where $\alpha_{k,i} : \gD \times \gD \mapsto \R_+$  are unknown, non-negative selection ``increment''
mappings, and
$b_i: \gU \times \gS \mapsto \R_+$ are unknown upper bounds on popularity bias (for $1 \leq i \leq M, k \in \sN$). Let $\bs{b}(u, \bs{s}) = \brk*{b_1(u, \bs{s}), \hdots, b_M(u, \bs{s})}^T$, $b_{\max} = \max_{1 \leq i \leq M, u \in \gU, \bs{s} \in \gS} b_i(u, \bs{s})$, and $\alpha_{\min} = \min_{I \in \gD, 1 \leq i \leq M,t \in \sN} \alpha_{t,i}(I, I)$.
\end{assumption}
One popularity bias metric is simply the number of selections, i.e., $\popularitybias_{t}\brk*{u, \bs{s},\bs{c}}_i \propto n_t(I_i)$, given by increments ${\alpha_{k,i}(I, c) = \alpha_0(I) \cdot \indicator{I = c}}$. Such bias often arises when the number of selections, views, likes, etc. of items in a slate are shown to the user.\footnote{More generally, popularity bias may be influenced by \emph{exogenous} factors, and could also be a \emph{stochastic} function of the selection history.} In what follows, we assume $b_{\max} < \infty$ and $\alpha_{\min} > 0$.

Next, we model quality and popularity bias using parametric functions. Specifically, each user $u$ and slate $\bs{s}$ is associated with an embedding vector $\bs{x}_q(u,\bs{s}) \in \R^{d_q}$ and $\bs{x}_p(u, \bs{s}) \in \R^{d_p}$. In this work we focus on the disentangling of quality and popularity, and assume that rank bias $\rankbias$ is known (though it still affects user disposition/choice). 
\begin{assumption}[Parametric Form]
\label{assump: parametric form}
For $1 \leq i \leq M$ and any $u \in \gU$, $\bs{s} \in \gS$, we assume $\qualitybias\brk*{u, \bs{s}}_i = \bs{x}_q^T(u,\bs{s}) \bs{\theta}^*_i$ and $b_i(u, \bs{s}) = \bs{x}_q^T(u,\bs{s})\bs{\phi}^*_i$, where $\bs{\theta}^*_i \in \R^{d_q}, \bs{\phi}^*_i \in \R^{d_p}$ ($i\leq M$) are unknown parameter vectors.
\end{assumption}

Such formulations are standard in the contextual bandit literature (e.g., \citep{faury2020improved,abeille2021instance,amani2021ucb}). Note that slate-based features (e.g., $\bs{x_q}(u, \bs{s})$) generalize per-item features.
To see this, let $\bs{s} = \brk*{I_1, \hdots, I_M}$, and let $\bs{q}(u,\bs{s}) = \brk*{\tilde{\bs{q}}^T_u(I_1), \hdots, \tilde{\bs{q}}^T_u(I_M)}^T \in \R^{dM}$, where $\tilde{\bs{q}}^T_u(I_M) \in \R^d$ are per-item features. We obtain $\brk[s]*{\qualitybias\brk*{u, \bs{s}}}_i = \tilde{\bs{q}}_u(I_i)^T \bs{\theta}^*_i$. 

We also note that $b_i(u,\bs{s})$ of \Cref{assump: parametric form} relates to the ``saturated" popularity bias of \Cref{assump: popularity bias}. In fact, we do not assume any specific functional form over the transient, non-stationary popularity bias, as defined in \Cref{assump: popularity bias} (apart from monotonicity).

Finally, we adopt a standard boundedness assumption from the GLM and multinomial bandit literature for the parameters $\bs{\theta}$ and $\bs{\phi}$. 
\begin{assumption}[Parameter Boundedness]
\label{assump: param boundedness}
For any $1 \leq i \leq M$, $\norm{\bs{\theta}^*_i}_2 \leq L_q$, $\norm{\bs{\phi}^*_i}_2 \leq L_p$.
\end{assumption}
Our model's assumptions are tailored to highlight challenges in RSs with popularity bias while ensuring analytical feasibility. Of note is \Cref{assump: popularity bias}, which underscores the typical observation that popular items tend to gain more traction over time. While our model makes this specific assumption, it aligns well with empirical results from prior work \citep{zhu2021popularity,zhang2021causal,zhu2022evolution}. Our main contribution is laying the theoretical groundwork to better understand the impact of popularity bias on recommender systems.

In the next section we derive theoretical lower and upper bounds for ranking with popularity-bias, and provide a first provably efficient algorithm for mitigating popularity bias in RSs.

\section{Overcoming Popularity Bias}
\label{section: overcoming popbias}

We now turn to the problem of designing optimal ranking policies in the presence of popularity-biased users. We first show that, in the general case, the problem is non-identifiable (i.e., one cannot disambiguate popularity bias from quality). In turn, we prove that any algorithm induces linear-regret, providing explicit lower bounds. This suggests further modeling assumptions are needed to derive efficient algorithms with sub-linear regret. 

To achieve this, we devise a diversity criterion w.r.t.\ the user population. We then derive an exploration-explotation UCB-style algorithm which effectively decouples popularity and quality. To that end, we prove efficient regret bounds for our algorithm. Our analysis puts emphasis on the non-trivial requirement to decouple the confounding effect of quality and popularity, and the attention to exploration needed, to obtain efficient guarantees. For the rest of this section we focus on the case of $\utility \equiv \qualitybias$.

\subsection{Nonidentifiability of Qualities}
\label{sec: lower bound}

A natural approach to maximizing value involves the ranker estimating item qualities ($\qualitybias$). However, as the ranker only observes user selections, doing so requires disentangling the effect of popularity bias on observed behavior from that of quality and rank bias. Unfortunately, this problem is generally non-identifiable, as we demonstrate with a simple counterexample.

Consider the set of one-dimensional problems with dimension $d=1$, slate size $M=1$, and two items $\gD = \brk[c]*{I_1, I_2}$, with qualities $q(I_1) = 1$, $q(I_2) = -1$, and no rank bias ($\rankbias = 0$). The two problems differ only in their item features. Concretely, in Problem~1 we let $\theta^* = \epsilon, b(I_1) = 1-\epsilon, b(I_2) = 1+\epsilon$, and in Problem~2, $\theta^* = -\epsilon, b(I_1) = 1+\epsilon, b(I_2) = 1=\epsilon$, for some $\epsilon > 0$ (see \Cref{fig: impossibility} for an illustration). 

By monotonicity of popularity (\Cref{assump: popularity bias}), whenever the ranker presents item $I_i$ and it is selected by a user, its popularity bias increases by at least $\alpha_{\min}$. After $C\frac{b(I_i)}{\alpha_{\min}}\log\brk*{\frac{2}{\delta}}$ selections of $I_i$, with high probability, $n_t(I_i) \geq \frac{b(I_i)}{\alpha_{\min}}$. When this event occurs at time, say, $t_0$, popularity bias reaches ``saturation" (i.e., $b(I_i)$).
Then for any time $t \geq t_0$, $\popularitybias_{t}(u, I, \bs{c}) = b(u, I)$, and the selection probability of selecting $I_i$ in both Problems~1 and ~2 is identical:
\begin{align*}
    z_i(\qualitybias\brk*{u, I_i} + \popularitybias_t\brk*{u, I_i, \bs{c}}) 
    =
    z_i(\pm \epsilon + (1 \mp \epsilon))
    =
   e/(1-e).
\end{align*}
At this point, w.h.p., the selection distribution in both problems is the same. Notably, the distribution remains the same for all $t \geq t_0$; hence, the ranker is unable to differentiate quality from popularity further, resulting in linear expected regret. Formally, we have the following result (see proof in \Cref{appendix: impossibility proof}):
\begin{theorem}[Impossibility]
\label{theorem: lower bound}
    For any ranking algorithm $\gR$, there exists a ranking problem for which the expected regret is lower bounded by $\E{\text{Reg}_{\gR}(T)} \geq \Omega(T)$. 
\end{theorem}
We note that a similar problem can arise even if $\popularitybias$ does not saturate, as the softmax function $z_i$ itself reaches saturation with increasing popularity, rendering estimation  exponentially hard \citep{amani2021ucb}. In the remainder of this section, we take steps to mitigate the identifiability problem in order to achieve sublinear regret. We begin by defining a ranker that is optimal when popularity has reached saturation. Then, motivated by the lower bound in \Cref{theorem: lower bound}, we show how popularity and quality can be disentangled through a variability assumption on the quality and popularity features, and construct an efficient, UCB-style exploration algorithm.

\subsection{The Quality Ranker}
\label{sec:quality_ranker}

Learning the optimal ranker $\gR^\ast$ requires good estimation of popularity bias dynamics in order to plan effectively. However, since the Markov process underlying our model is non-recurrent (due to ever-increasing popularity), estimation of these dynamics is generally not possible without additional recurrence assumptions (e.g., the ability to reset the environment). To work around this issue, we first consider a baseline ranking policy which assumes stationarity of the Markov process. The stationary \emph{quality ranker} $\gR_q$ recommends slates as follows: 
\begin{align*}
    \bs{s}^q(u) 
    \in 
    \arg\max_{\bs{s} \in \gS}
    \sum_{i=1}^M
    \brk[s]*{\qualitybias(u, \bs{s})}_i 
    \cdot
    z_i(\qualitybias(u, \bs{s}) + \bs{b}(u, \bs{s}) + \rankbias).
\end{align*}
Notice that, when $b_i$ is not slate dependent (i.e., $b_i(u, \bs{s}) \equiv b_i(u)$), the slate recommended by
$\gR_q$ 
comprises
the $M$ highest-quality items in decreasing order.

The quality ranker is, of course, suboptimal in the general case. However, it is optimal in a counterfactual world where all popularity biases have saturated (hence, in which a steady-state distribution of the Markov chain has been reached).
The following results shows that the quality ranker $\gR_q$ is in fact \emph{asymptotically optimal}. Specifically, for large enough $T$, it achieves constant regret w.r.t.\ the optimal ranker.
\begin{theorem}[Asymptotic Optimality of the Quality Ranker]
\label{theorem: asymptotic optimality}
   Let $\delta \in (0,1)$ and assume $\alpha_{\min} > 0$. For any $T \geq 1$, with probability at least $1-\delta$, the regret of $\gR_q$ is 
   \begin{align*}
       \text{Reg}_{\gR_q}(T) \leq
       \gO\brk*{\frac{\abs{\gD}Mb_{\max}}{\alpha_{\min}}\log\brk*{\frac{\abs{\gD}}{\delta}}}.
   \end{align*}
\end{theorem}
A direct corollary of \Cref{theorem: asymptotic optimality} is that the regret of any ranker $\gR$ can be written as
$
    \text{Reg}_{\gR}(T) 
    \leq  
    \gO\brk*{v^{\gR_q}_1 - v^{\gR}_1 + \frac{\abs{\gD}Mb_{\max}}{\alpha_{\min}}\log(\frac{\abs{\gD}}{\delta})}.
$
This motivates the analysis of $v^{\gR_q}_1 - v^{\gR}_1$ (the regret w.r.t. $\gR_q$), as both are equivalent for large enough $T$.

\subsection{Identifiability Through Variability}
\label{sec:variability}

We showed in \Cref{sec: lower bound} that popularity and quality biases cannot be disentangled in the general case. We now show that, under a condition of sufficient variability induced by the user population, these two factors can be decoupled, a fact we exploit below to design a ranking policy with sublinear regret.

Let $\bs{\Sigma}_{\bs{q}\bs{p}}(u, \bs{s}) = \bs{x}_{\bs{q}}(u, \bs{s})\bs{x}_{\bs{p}}^T(u, \bs{s})$ denote the correlation matrix of $\bs{x}_{\bs{q}}$ and $\bs{x}_{\bs{p}}$. Similarly, we use the notations $\bs{\Sigma}_{\bs{q}\bs{q}}, \bs{\Sigma}_{\bs{p}\bs{p}}$. That is,  $\expect*{u}{\bs{\Sigma}_{\bs{q}\bs{p}}(u, \bs{s})} = \expect*{u}{\bs{x}_{\bs{q}}(u, \bs{s})\bs{x}_{\bs{p}}^T(u, \bs{s})}$ captures the correlation of features in $\qualitybias$ and $\bs{b}$, in expectation over the user population. We make the following assumption to ensure identifiability of quality:
\begin{assumption}
\label{assumption: lambda min}
    For all $t \leq T$ and $\bs{s}_t = \bs{s}_t(u_t)$, there exists $\rho \in (0, 1)$, such that 
    \begin{align*}
        \expect*{t-1}{
        \begin{pmatrix}
            \rho \bs{\Sigma}_{\bs{q}\bs{q}}(u_t, \bs{s}_t) & \bs{\Sigma}_{\bs{q}\bs{p}}(u_t, \bs{s}_t) \\ \bs{\Sigma}_{\bs{p}\bs{q}}(u_t, \bs{s}_t) & \bs{\Sigma}_{\bs{p}\bs{p}}(u_t, \bs{s}_t)
        \end{pmatrix}
        }
        \succeq
        0.
    \end{align*}
    Let $\rho_{\min} \in (0,1)$ be the smallest such $\rho$.
\end{assumption}
Similar assumptions have been made in contextual bandit models \citep{kannan2018smoothed,chatterji2020osom,bastani2021mostly,papini2021leveraging}. Nevertheless, \Cref{assumption: lambda min} is, in fact, less demanding than these previous assumptions, requiring only positive-semi-definiteness (see \Cref{appendix: lambda min assump} for explicit comparison to other assumptions). Also, notice that \Cref{assumption: lambda min} holds trivially for $\rho = 1$, yet we require it to hold for some $\rho < 1$. The assumption suffices to disentangle quality from popularity by ensuring enough variability exists in the popularity features. Importantly, this assumption enables the design of an efficient algorithm which accounts for popularity bias, as we show next.


\begin{algorithm}[t!]
\caption{QP Ranker: $\gR_{qp}$}
\label{alg: qrp-ucb}
\begin{algorithmic}[1]
\STATE{ \textbf{require:} $\delta \in (0,1), \lambda > 0$} 
\STATE{ \textbf{initialize:} $\gH \gets \emptyset, \bs{V} \gets \lambda \bs{I}, \tau_{\min} = \frac{8M b_{\max}}{\alpha_{\min}}\log(\frac{\abs{\gD}}{\delta})$ }
\FOR{$t = 1, 2, \hdots$}
    \STATE Observe user $u_t \sim P_{\gU}$
    \STATE $\hat{\bs{\psi}}^{\text{ML}} \in \arg\max_{\bs{\psi}}\gL_\lambda(\bs{\psi}| \gH)$
    \STATE $\hat{\bs{\psi}} \in \arg\min_{\bs{\psi}} \norm{g\brk*{\bs{\psi} | \gH} - g(\hat{\bs{\psi}}^{\text{ML}} | \gH)}_{\bs{V}^{-1}}$
    \STATE ${\bs{s}_t \in \arg\max_{\bs{s} \in \gS} \sum_{i=1}^M
    \brk[s]*{\qualitybias_{\hat{\bs{\theta}}}(u_t, \bs{s})}_i z_i\brk*{\disp_{\hat{\bs{\theta}},\hat{\bs{\phi}}}(u_t, \bs{s})} 
    + \epsilon_{t,\delta}(\bs{s})}$
    \STATE Select slate $\bs{s}_t$ and observe selected item $c_t$
    \IF {$\abs{\brk[c]*{I \in \bs{s}_t : n_t(I) \geq \tau_{\min}}} = M$}
        \STATE $\gH \gets \gH \cup (u_t, \bs{s}_t, c_t)$
        \STATE $\bs{V} \gets \bs{V} + \bs{\Sigma}_x(u_t, \bs{s}_t)$
    \ENDIF
\ENDFOR
\end{algorithmic}
\end{algorithm}

\begin{figure*}[t!]
    \centering
    \begin{subfigure}[b]{0.35\textwidth}
    \centering
    \includegraphics[clip, trim=0cm 0 0 0, width=\linewidth]{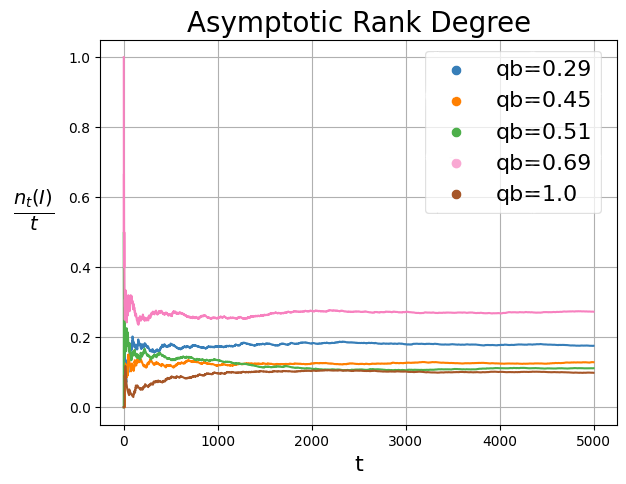}
    \captionsetup{margin={1cm,0cm}} 
    \caption{Unlucky Draw}
    \end{subfigure}
    \begin{subfigure}[b]{0.35\textwidth}
    \centering
    \includegraphics[clip, trim=0cm 0 0 0, width=\linewidth]{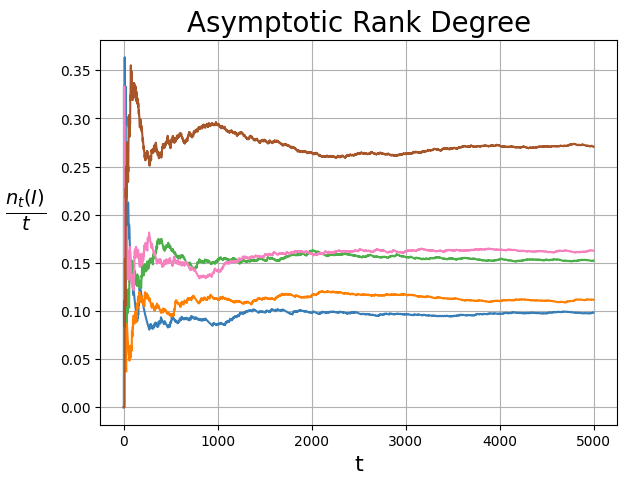}
    \captionsetup{margin={1cm,0cm}} 
    \caption{Lucky Draw}
    \end{subfigure}
    \hspace{3em}
    \caption{The empirical selection probabilities of different items (of varying qualities, indicated by color) by a popularity-driven ranker as a function of time. (a) An ``unlucky'' draw where unfavorable effects of the popularity feedback loop are evident: low-quality items $I$ move into higher positions on the slate and remain there, and hence have high asymptotic selection probability $\lim_{t\to\infty} n_t(I)/t$. (b) A ``lucky'' draw in which high-quality items move to the top of the slate and hence have high asymptotic selection probability.}
    \label{fig:powerlaws}
\end{figure*}

\subsection{The QP Ranker} 
\label{section: QP ranker}

We turn to the specification and analysis of our key algorithm, the \emph{Quality-Popularity (QP) Ranker} $\gR_{qp}$, presented in \Cref{alg: qrp-ucb}. Our approach exploits the identifiability condition in \Cref{assumption: lambda min} to disentangle popularity from quality.

For convenience, we let $\bs{\psi} \in \R^{d_q + d_p}$ be the concatenation $\bs{\psi}=\brk*{\bs{\theta}_1^T, \hdots, \bs{\theta}_M^T, \bs{\phi}_1^T, \hdots, \bs{\phi}_M^T}^T$ of the quality and popularity bias parameters, and $\bs{x}(u,\bs{s}) \in \R^{d_q + d_p}$ the concatenation $\bs{x}=\brk*{\bs{x_q}^T, \bs{x}^T_p}^T$ of the quality and popularity embeddings. This notation admits the following shorthand: $\qualitybias(u,\bs{s}) + \bs{b}(u, \bs{s}) = \bs{x}^T(u, \bs{s})\bs{\psi}$.

We begin by defining a projected, regularized likelihood estimator which we use in \Cref{alg: qrp-ucb}. With the convention that $I_{0, k} = \emptyset$ for all $k \leq T$, we define the regularized log-likelihood with regularization parameter $\lambda > 0$ by:
\begin{align}
\label{eq: likelihood}
    \gL_\lambda\brk*{\bs{\psi}| \gH_t}
    \! = \!\!\!\!\!\!\!\!
    \mathop{\sum_{u_k, \bs{s}_k, c_k \in \gH_t}}_{1\leq i\leq M} \!\!\!\!\!\!\!\!
    \indicator{c_k = I_{i,k}}\!
    \log\brk*{\!z_i\brk*{\!\disp_{\bs{\theta}, \bs{\phi}}(u_k,\! s_k)}}
   \! - \!\!
    \frac{\lambda}{2}\norm{\bs{\psi}}_2^2,
\end{align}
where
$
    \disp_{\bs{\theta}, \bs{\phi}}(u_k, s_k)
    =
    \qualitybias_{\bs{\theta}}(u_k, \bs{s}_k) \! +\! \popularitybias_{\bs{\phi}}(u_k, \bs{s}_k) \! +\! \rankbias.
$

We then define the maximum likelihood estimator (MLE) by 
\begin{align*}
\hat{\bs{\psi}}_t^{\text{ML}} \in \arg\max_{\bs{\psi}}\gL_\lambda\brk*{\bs{\psi}| \gH_t}.
\end{align*}

Next, to ensure efficient exploration, we require the estimation of confidence bounds for the quality and popularity bias parameters. For this, we let $\bs{V}_t = \lambda \bs{I} + \sum_{u_k, \bs{s}_k, c_k \in \gH_t} \bs{\Sigma}_x(u_k, \bs{s}_k)$ be the design matrix, and define
\begin{align*}
    \bs{g}\brk*{\bs{\psi}| \gH_t}
    =
    \bs{\psi}
    +
    \sum_{u_k, \bs{s}_k, c_k \in \gH_t}
    \bs{z}\brk*{\disp_{\bs{\theta}, \bs{\phi}}(u_k, s_k)}
    \otimes
    \bs{x}(u_k, \bs{s}_k).
\end{align*}
The function $\bs{g}$ is a key quantity related to the gradient of $\gL_\lambda$, which allows us to achieve tight confidence guarantees for $\bs{\psi}$. We define the projected MLE at time $t$ by
\begin{align*}
    \hat{\bs{\psi}} \in \arg\min_{\bs{\psi}} \norm{g\brk*{\bs{\psi}| \gH_t} - g\brk*{\hat{\bs{\psi}}_t^{\text{ML}| \gH_t}}}_{\bs{V}_t^{-1}}.
\end{align*}
We note that, though tighter estimators exist (see e.g., \citet{amani2021ucb}), the above estimator is sufficient for efficient regret guarantees.

We are now ready to define the QP-Ranker ($\gR_{qp}$) in \Cref{alg: qrp-ucb}. The QP-Ranker proceeds in discrete steps. When a user is sampled at time~$t$, $\gR_{qp}$ uses the current MLE $\hat{\bs{\psi}}$ to construct its slate, using an additive bonus $\epsilon_{t,\delta}(\bs{s})$. The design matrix $\bs{V}$ is updated only if the popularities of all items in the slate have reached saturation. Let $\epsilon_{t,\delta}(\bs{s}) = \brk*{4\sqrt{M} + \frac{1}{\sqrt{1-\rho_{min}}}}\gamma_t(\delta)
        \sqrt{
        \expect*{t-1}{
            \norm{
                \bs{x}(u_t, \bs{s})
            }_{\bs{V}_t^{-1}}^2
        }}$, where $\gamma_t(\delta) = 4e^4M^2\brk*{\sqrt{\lambda M}L + 2 \sqrt{\log\brk*{\frac{1}{\delta}} + Md\log\brk*{1+\frac{t}{\lambda d}}}}$, and we used the notation $d = d_q + d_p$ and $L = L_q + L_p$. We have the following result (its proof is given in \Cref{appendix: regret proof}).

\begin{theorem}
\label{thm: regret upper bound}
    Let $\delta \in (0,1)$, and $\lambda =\frac{1}{\sqrt{L}}$. Then with probability at least $1-\delta$, the regret of the QP Ranker $\gR_{qp}$ (\Cref{alg: qrp-ucb}) is upper bounded by
    \begin{align*}
        \text{Reg}_{\gR_{qp}}(T) \leq 
        \gO\brk*{M^{2.5}
            d\sqrt{T}
            \brk*{\sqrt{M} + \frac{1}{\sqrt{1-\rho_{min}}}}
             \log\brk*{1 + \frac{TL}{d}}
             \sqrt{\log\brk*{\frac{1}{\delta}}}
            }.
    \end{align*}
\end{theorem}

\Cref{thm: regret upper bound} shows sublinear regret is achievable with effective exploration and proper accounting of popularity bias. Still, we note that \Cref{assumption: lambda min} is needed to ensure user utility (quality) can be disentangled from popularity, per our lower bound in \Cref{theorem: lower bound}. Particularly, we see the effect of $\rho_{\min}$ on overall regret, which acts as a hardness parameter for the problem.  We note too that the $\gO$-notation in \Cref{thm: regret upper bound} hides the constant of \Cref{theorem: asymptotic optimality}, which may be large, but does not depend on $T$. Future research should address this using, say, further assumptions on the dynamics of popularity.

The QP Ranker also performs effectively in simulation, as we demonstrate in the following section. We also conducted several additional studies of the algorithm, including varying problem parameters and assessing instances where \Cref{assumption: lambda min} does not hold. We refer the reader to \Cref{appendix: experiments} for these experiments. 

\begin{figure*}[t!]
    \centering
    \begin{subfigure}[b]{0.35\textwidth}
    \centering
    \includegraphics[width=\linewidth]{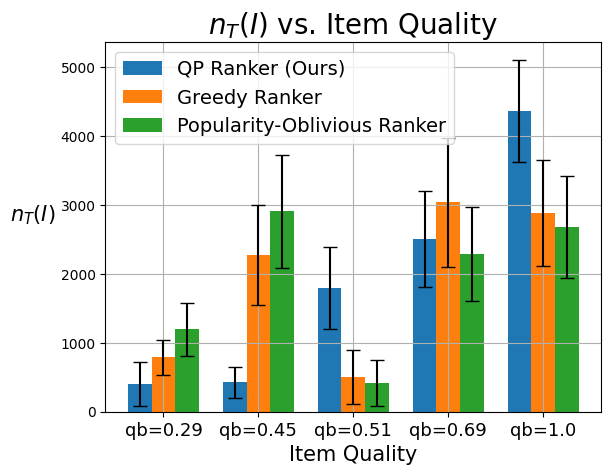}
    \end{subfigure}
    \hspace*{1cm}
    \begin{subfigure}[b]{0.35\textwidth}
    \centering
    \includegraphics[width=\linewidth]{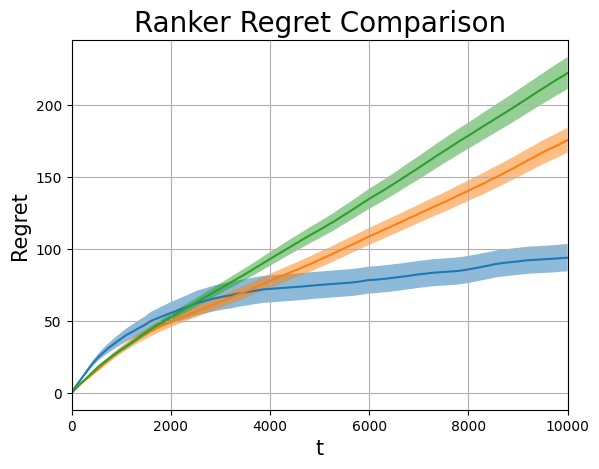}
    \end{subfigure}
    \caption{Comparison of the QP, greedy and popularity-oblivious rankers on popularity-biased users. Left: number of selections per item (items of different qualities). Right: overall regret. We see that not accounting for user popularity bias increases overall regret significantly; by contrast our QP ranker has low regret and tends to rank high-quality items higher (inducing more high-quality selection).}
    \label{fig: regret biased users}
\end{figure*}

\section{Simulation Studies}
\label{sec:experiments}

In this section we conduct several simulation studies to demonstrate the efficiency of the QP-Ranker, and emphasize the drawbacks of (i) greedy rankers, which do not explore efficiently, (ii) popularity biased rankers (see \Cref{section: pop-biased rankers}), and (iii) popularity-oblivious rankers, which ignore users' popularity bias. For this, we adopt a simple realization of our framework. We assume $N$ items, with user-item embeddings $\bs{x}(u, I)$ and parameters $\bs{\theta}^*$ sampled uniformly on $[0, 1/\sqrt{d}]^d$. This gives a slate embedding $\bs{x}(u, \bs{s}) \in \R^{Md}$.

We start by considering the popularity-driven ranker, which only takes into account an item's popularity when ranking. Concretely, at any time $t$, it generates a slate $\bs{s}_t$ whose $i^{\text{th}}$ item $I_{i,t}$ has maximal count $n_t(I)$, excluding those items $I_{j,t}$ where $j<i$. Formally, let $I_{i,t}$ be the item at position $i$ in slate $\bs{s}_t$. Then,
\begin{align*}
    I_{i,t} \in \arg\max_{I \in \gD, I \notin \brk[c]*{I_{j,t}}_{j=1}^{i-1}} n_t(I).
    \tag{popularity-driven Ranker}
\end{align*} 

\Cref{fig:powerlaws} depicts the dynamics of the popularity-driven ranker. While user dispositions (hence selections) are positively correlated with quality, due to rank/position bias, the ranker gradually \emph{amplifies} the positions of higher-ranked items given their biased selections. Of course, the ranker may get lucky---if higher quality items are selected more often in early stages of the process, the ranker may converge to high-utility slates. However, \Cref{fig:powerlaws} shows that this process may also converge to highly suboptimal rankings, due to the feedback loop between user rank bias in selection and the ranker's own popularity bias. This, in turn, induces linear regret and low user welfare. We refer the reader to \Cref{appendix: dynamics of popularity biased rankers} for theoretical analysis of these dynamics.

Next, we analyze and compare greedy (i.e., non-exploratory) as well as popularity-oblivious rankers in the setting of popularity-biased users. Specifically we consider two baseline rankers. The first ``greedy'' ranker takes into account user popularity, but does not attempt to explore. We include this ranker to emphasize the need to efficiently explore in order to minimize regret and maximize overall user welfare. Particularly, the greedy ranker is defined by \Cref{alg: qrp-ucb} with the exploration coefficient $\epsilon$ set to zero. The second we consider is the popularity-oblivious ranker. In this case, popularity bias is completely ignored, and the ranker attempts to efficiently explore using a misspecified user selection model. Particularly, the oblivious ranker is defined by \Cref{alg: qrp-ucb} with $\popularitybias$ set to zero.

\Cref{fig: regret biased users} compares the QP ranker, the popularity-oblivious ranker and the greedy ranker in terms of regret and item selections. It is evident that the oblivious ranker has significantly worse regret than our QP ranker due to misspecification, and the greedy ranker has worse regret due to its failure to explore efficiently. \Cref{fig: regret biased users} also  shows the number of selections of each item at the end of training (each is labeled with its quality). The popularity-oblivious ranker tends to rank low quality items higher, inducing more selections of such items via the dynamics of amplification.

Finally, we note that the popularity-oblivious ranker may actually perform reasonably well in the case where user utility is dependent on both item quality \emph{and} popularity, e.g., if $\utility \equiv \qualitybias + c \cdot \popularitybias$ for some $c > 0$. Such a case might arise when a user's utility is dictated not only by inherent item quality, but also by positive ``network effects'' that are correlated with popularity (e.g., the number of friends with whom the user can discuss a recent movie). That said, the oblivious ranker is still not efficient w.r.t.\ this combined utility, since disambiguation of popularity and quality and explicit exploration are still required.



\section{Related Work}
\label{sec:related}

Our work intersects with various lines of research, including the investigation of popularity bias in RSs, contextual bandits, and power-law distributions in RSs.

\paragraph{Popularity Bias in RSs.}
Popularity bias in recommender systems has garnered significant focus in the research community. Prior studies have examined the trade-off between item popularity and recommendation accuracy \citep{steck2011item}, as well as the significance of long-tail items in elevating user satisfaction and averting monopoly by dominant brands \citep{park2008long}. In a more encompassing manner, work by \citet{jannach2015recommenders} empirically exhibited the differing susceptibilities of various recommendation algorithms to popularity bias. This phenomenon has been thoroughly studied in the literature, with multiple approaches to define, evaluate, and counteract this bias \citep{janssen2010rank,abdollahpouri2019popularity,abdollahpouri2019unfairness,elahi2021investigating,abdollahpouri2021user}. Some of the approaches proposed to alleviate the detrimental effects of this bias indclude pre-processing of the training data \citep{jannach2015recommenders,bellogin2017statistical}, model-based and re-ranking methodologies \citep{abdollahpouri2020popularity}, the inclusion of novelty scores \citep{bedi2014using}, and the use of Variational Autoencoders \citep{borges2021mitigating}. Calibrated Popularity \citep{abdollahpouri2021user} and time-aware recommender systems \citep{campos2014time,harshvardhan2022ubmtr} provide user-centric and temporal insights, respectively. Recent work has also delved into human-in-the-loop (HitL) bias in conversational RSs \citep{fu2021hoops,fu2021popcorn}. While much of these works provide strong empirical evidence for the existence and possible mitigation strategies of popoularity bias in RSs, our work provides the first theoretical lower and upper bounds showing impossibility results for identifiability of quality from popularity, and providing a provably efficient ranking algorithm from maximizing user-welfare under popularity bias, with strong regret guarantees.

\paragraph{Contextual Bandits.} Our work is related to the growing literature of generalized linear bandits \citep{filippi2010parametric}, where rewards are sampled from a logistic function with linear features \citep{abeille2021instance,amani2021ucb}. This also includes the expansion of the logistic formulation to the $M$-dimensional case through multinomial bandits \citep{amani2021ucb}. By extending these models, we account for dynamic popularity bias and user utility, while also proposing a novel user-choice model, inspired by practical concurrent learning models \citep{tennenholtzreinforcement,tennenholtz2023reinforcement}. This brings to light the challenges of confounding effects in bandits, where popularity bias may act as a confounder for item quality or user utility \citep{bareinboim2015bandits,krishnamurthy2018semiparametric,tennenholtz2021bandits}.

\paragraph{Power Laws and Preferential Attachment.}
Connections between our model and classical mechanisms proposed for the emergence of power laws are evident. For instance, under certain assumptions, our model aligns with the principles of preferential attachment that results in Yule-Simon size distributions \citep{mitzenmacher2004brief}. Additionally, other configurations of our model resemble the prestige ranking model \citep{janssen2010rank}, emphasizing the selection distribution that leads to the manifestation of power laws.

\section{Conclusion}

We proposed theoretically grounded methods and regret analysis of popularity bias in RSs. Our work provides insights into the impact of popularity bias on user welfare and proposes methods to mitigate negative welfare impact. Our results strongly suggest that popularity bias can negatively impact user welfare, especially in naive RSs that use popularity as a proxy for quality. We show that exploration can help mitigate these negative effects, leading to better long-term user utility (and sub-linear rather than linear regret). Moreover, our work highlights the importance of a deeper understanding of the \emph{mechanisms} underlying popularity bias in ranking-based systems. Given the significant role these play in shaping user opinions and decisions, this understanding is crucial for ensuring fairness, diversity and user well-being. Our work contributes to this understanding and should have implications for the design and deployment of ranking systems.

Our framework has several limitations. Primarily, our analysis assumes that user utility depends solely on selected item quality; however, other factors (e.g., social, cultural) may shape utility. We also note that popularity bias is prevalent in other ranking-based systems (e.g., search engines, social media platforms, news aggregators). Our findings should have implications in these cases, but further research is needed to determine their generalizability.

\bibliography{bibliography}
\bibliographystyle{plainnat}

\newpage
\onecolumn
\appendix

\section{Dynamics of the Popularity-Driven Ranker}
\label{appendix: dynamics of popularity biased rankers}

\begin{figure}[t!]
\centering
\includegraphics[width=0.8\linewidth]{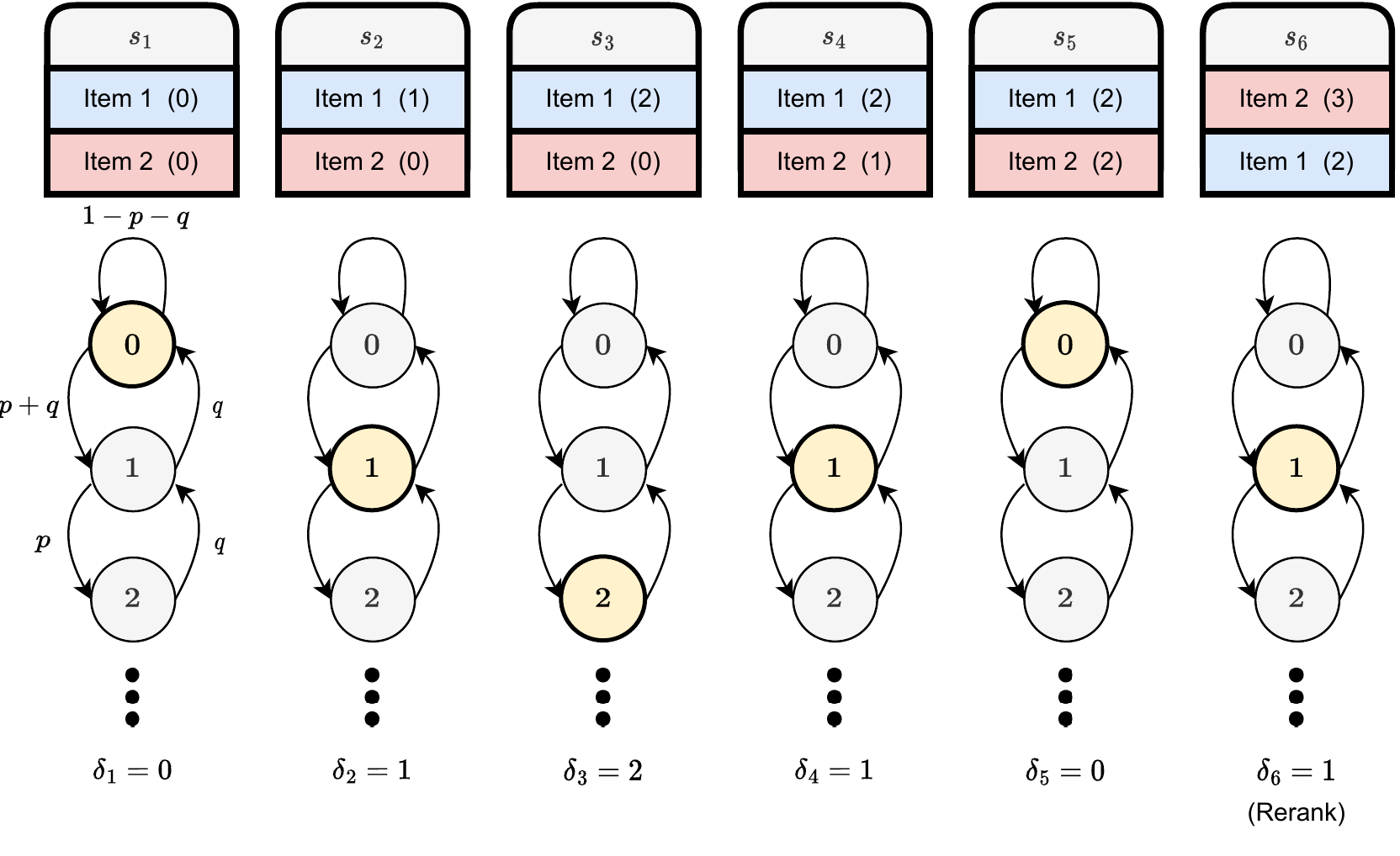}
\caption{\small A visualization of the position-biased ranking process for the case of two items. The number of selections on each item is shown in parenthesis. A rerank event occurs whenever the random walk $\delta_t$ hits the reflective barrier and exits it through a selection on the second position.}
\label{fig: reranking}
\end{figure}

We consider popularity bias that is directly incorporated by suboptimal rankers. For clarity, we focus on the simple case in which users are only biased by position rank (i.e., unbiased by quality or popularity). That is, we assume
$
    \qualitybias \equiv 0, \popularitybias \equiv 0
$,
such that a user selects position $1 \leq i \leq M$ w.p. $p_i := z_i\brk*{\rankbias}$, where $p_1 \geq p_2 \geq \hdots \geq p_M$. 

To this end, we wish to anaylze the position dynamics that could emerge, as a ranker accumulates its own popularity bias, and show it may eventually induce a power law distribution over items. For this, we consider a suboptimal, popularity-driven ranker, $\gR^{\text{pop}}$, which ranks items based on the number of times they were selected. Specifically, the ranker selects the top $M$-selected items and ranks them in order of number of selections. Let $S^{\text{pop}}_t$ denote the slate induced by the popularity ranker, and let $I_{i,t}$ be the item in position $i$. Then,
\begin{align*}
    I_{i,t} \in \arg\max_{I \in \gD, I \notin \brk[c]*{I_{j,t}}_{j=1}^{i-1}} n_t(I).
    \tag{popularity-driven Ranker}
\end{align*}
For consistency, we assume the ranker break ties by not reranking any two items if they have the same number of selections. We also assume the initial slate of items is chosen uniformly from the corpus. 

The above implies a Markov process, $(\bs{n}_t)_{t\in\mathbb{N}}$, whose state is defined by the number of selections. In the rest of this section, we will analyze this process to better understand its behavior and the possible emergence of power law distributions. Let $n_t(I_{r,t})$ be the rank degree -- a random variable which indicates the number of selections on the item at position $r$ in the ranking at time $t$. The rank-size distribution $\pi_\infty(r)$ is defined as the fraction of selections of the item at position $r$ at the limit, i.e., $\pi_\infty(r) = \lim\sup_{t\rightarrow \infty} \frac{1}{t}\expect*{}{n_t(I_{r,t})}$. Indeed, the rank-size distribution $\pi_\infty$ is a distribution over rank positions. In what follows, we will show this distribution is monotonically decreasing with rank, whenever users are more inclined to select higher ranked items.

Consider the case of two items, $\gD = \brk[c]*{I_1, I_2}$. We denote the probability of selecting positions $1$ and $2$ by $p$ and $q$, respectively. Notice that 
\begin{align}
\label{eq: rank degree decomposition}
n_t(I_{r,t})
= 
\sum_{k=1}^{t-1} \indicator{\text{position $r$ was selected at time $k$}} + \sum_{k=1}^{t-1} \indicator{\text{rerank occured at time $k$}}.
\end{align}
This comes from the fact that the item at position 1 and the item at position 2 can only swap positions if they both have the same size and the item in position 2 was selected. This will lead to the item at position 2 having a larger size and moving to position 1, hence the size at position 1 will increase by 1 even though position 2 was selected. 

By \Cref{eq: rank degree decomposition}, the rank degree is composed of the number of selections on position $r$, and the total number of reranks. Notice that the expected number of reranks at position $r$ at time $t$ is $p_r t$. To calculate the number of reranks, we denote the random process $\delta_t = n_t(I_{1,t}) - n_t(I_{2,t})$. This quantity evolves according to the following transition kernel 
\begin{align*}
    \delta_{t+1}
    =
    \indicator{\delta_t > 0}
    \begin{cases}
        \delta_t + 1  &,\text{w.p. } p \\
        \delta_t - 1 &, \text{w.p. } q \\
        \delta_t &, \text{w.p. } 1-p-q
    \end{cases}
    +
    \indicator{\delta_t = 0}
    \begin{cases}
        1  &,\text{w.p. } p + q \\
        0 &, \text{w.p. } 1-p-q
    \end{cases}.
\end{align*}

Notice that $\brk*{\delta_t}_{t \in \sN}$ is a homogeneous random walk on the one-dimensional line, with a reflective barrier at $0$. A reranking event occurs when, starting at $\delta_1 = 1$, the random walk hits the reflective barrier and then exits it through a selection on the second position (see \Cref{fig: reranking}). That is, 
\begin{align}
\label{eq: rerank probability}
P(\text{rerank at time $k > t$} | \delta_t = 1) 
= 
\underbrace{P(\exists k > t, \delta_k = 0 | \delta_t = 1)}_{\text{hit reflective barrier if stated at $\delta_t = 1$}} \cdot \underbrace{q}_{\text{selection on second position}}.
\end{align}
Once a rerank occurs, the random walk restarts at $1$ and proceeds identically. 

It remains to characterize $P(\exists k > t, \delta_k = 0 | \delta_t = 1)$ in \Cref{eq: rerank probability}. From the Gambler's Ruin analysis \citep{epstein2012theory}, we have that
\begin{align*}
    P(\exists k > t, \delta_k = 0 | \delta_t = 1)
    &=
    \lim\limits_{N \to \infty}
    P(\text{$\delta_t$ hits $0$ before $N$} | \delta_0 = 1) \\
    &=
    \lim\limits_{N \to \infty}
    \begin{cases}
        \frac{\left(\frac{q}{p}\right) - \left(\frac{q}{p}\right)^N}{1 - \left(\frac{q}{p}\right)^N}
        &, p \neq q \\
        1-\frac{1}{N} &, p = q
    \end{cases}
    =
    \begin{cases}
        \frac{q}{p} &, p > q \\
        1 &, p = q
    \end{cases},
\end{align*}
where in the last equality we used the fact that $p \geq q$. It then follows from \Cref{eq: rerank probability} that the expected number of reranks at time $t$ is upper bounded by $\frac{1}{1-\frac{q^2}{p}} = \frac{p}{p-q^2}$, for $p > q$. Finally, using \Cref{eq: rank degree decomposition}, we conclude that, for $p > q$
\begin{align*}
\pi_\infty^1 
= 
\lim_{t\to\infty} \frac{1}{t} \expect*{}{\text{number of selections on p. 1 }+\text{ number of reranks}} 
=  
\lim_{t\to\infty} \frac{1}{t}\left( p\cdot t + \text{Const}\right) = p.
\end{align*}
Similarly, $\pi_\infty^2 = q$. This, in turn, means that the dynamic ranking process would eventually freeze on a particular ranking for which one item will always be ranked above another. In other words, as long as the number of reranks is finite, the process would eventually freeze. By symmetry of the process, the ranking will freeze on any slate $\bs{s} \in \gS$ and any order of items with equal probability. Notably, when $p = q$ the number of reranks is infinite, and $\pi_\infty^1, \pi_\infty^2$ diverge, as both items are reranked infinitely often.

\newpage
\section{Discussion}
\subsection{The Case of No Popularity Bias}
\label{appendix: reduction to mnl}

When no popularity bias exists, our model reduces to the multinomial model of \citet{amani2021ucb}. We note that \citet{amani2021ucb} uses tighter confidence sets for the parameters due to exponential dependence on coefficients relating to the gradient of the softmax function $\bs{z}$. For clarity and simplicity, we chose to bound $\qualitybias, \popularitybias, \rankbias$ in the interval $[0,1]$ as to avoid this dependence. Nevertheless, our results easily extend to general intervals, for which case one can apply the tighter parameter estimation guarantees of \citet{amani2021ucb} to reduce the exponential dependence (yet not eliminate it). While the proofs and derivations of our results would not change, a different estimator would need to be used, projecting to a more involved set.

\subsection{\Cref{assumption: lambda min}}
\label{appendix: lambda min assump}

To better understand \Cref{assumption: lambda min} we first note the limit case of $\rho = 1$. While the assumptions requires $\rho < 1$ (as also evident by our regret bound in \Cref{thm: regret upper bound}), the assumption always holds for $\rho = 1$, since it is reduced to
$
        \expect*{t-1}{
        \begin{pmatrix}
            \bs{\Sigma}_{\bs{q}\bs{q}}(u_t, \bs{s}_t) & \bs{\Sigma}_{\bs{q}\bs{p}}(u_t, \bs{s}_t) \\ \bs{\Sigma}_{\bs{p}\bs{q}}(u_t, \bs{s}_t) & \bs{\Sigma}_{\bs{p}\bs{p}}(u_t, \bs{s}_t)
        \end{pmatrix}
        }
        \succeq
        0,
$
which holds by definition. It follows that a sufficient condition for \Cref{assumption: lambda min} to hold for some $\rho < 1$ is:
\begin{align}
        \expect*{t-1}{
        \begin{pmatrix}
            \bs{\Sigma}_{\bs{q}\bs{q}}(u_t, \bs{s}_t) & \bs{\Sigma}_{\bs{q}\bs{p}}(u_t, \bs{s}_t) \\ \bs{\Sigma}_{\bs{p}\bs{q}}(u_t, \bs{s}_t) & \bs{\Sigma}_{\bs{p}\bs{p}}(u_t, \bs{s}_t)
        \end{pmatrix}
        }
        \succ
        0.
        \label{eq: sufficient assumption}
\end{align}
The assumption in \Cref{eq: sufficient assumption} has been used in previous work on contextual bandits (e.g., \citep{chatterji2020osom}). Nevertheless, we emphasize that \Cref{assumption: lambda min} can hold even in cases where \Cref{eq: sufficient assumption} does not hold. For example, consider a case where popularity bias does not exist. In this case, we have
\begin{align*}
    \expect*{t-1}{
        \begin{pmatrix}
            \bs{\Sigma}_{\bs{q}\bs{q}}(u_t, \bs{s}_t) & \bs{\Sigma}_{\bs{q}\bs{p}}(u_t, \bs{s}_t) \\ \bs{\Sigma}_{\bs{p}\bs{q}}(u_t, \bs{s}_t) & \bs{\Sigma}_{\bs{p}\bs{p}}(u_t, \bs{s}_t)
        \end{pmatrix}
    }
    =
    \expect*{t-1}{
        \begin{pmatrix}
            \bs{\Sigma}_{\bs{q}\bs{q}}(u_t, \bs{s}_t) & \bs{0} \\ 
            \bs{0} & \bs{0}
        \end{pmatrix}
    }.
\end{align*}
While clearly 
$\expect*{t-1}{
        \begin{pmatrix}
            \bs{\Sigma}_{\bs{q}\bs{q}}(u_t, \bs{s}_t) & \bs{0} \\ 
            \bs{0} & \bs{0}
\end{pmatrix}} \not\succ 0$, one can see that 
$\expect*{t-1}{
        \begin{pmatrix}
            \rho \bs{\Sigma}_{\bs{q}\bs{q}}(u_t, \bs{s}_t) & \bs{0} \\ 
            \bs{0} & \bs{0}
\end{pmatrix}} \succeq 0$, for all $\rho \in (0,1)$. Indeed, \Cref{assumption: lambda min} is less demanding than previous assumptions used in the contextual bandit literature. Moreover, the assumption is used only to disambiguate popularity from quality. To gain intuition as to why \Cref{assumption: lambda min} becomes stronger as $\rho \to 0$, notice it does not hold whenever $\rho = 0$ and $\expect*{t-1}{\bs{\Sigma}_{\bs{q}\bs{q}}(u_t, \bs{s}_t)}, \expect*{t-1}{\bs{\Sigma}_{\bs{p}\bs{p}}(u_t, \bs{s}_t)} \neq \bs{0}$, since the matrix
\begin{align*}
    \expect*{t-1}{
        \begin{pmatrix}
            \bs{0} & \bs{\Sigma}_{\bs{q}\bs{p}}(u_t, \bs{s}_t) \\ \bs{\Sigma}_{\bs{p}\bs{q}}(u_t, \bs{s}_t) & \bs{\Sigma}_{\bs{p}\bs{p}}(u_t, \bs{s}_t)
        \end{pmatrix}
    }
\end{align*}
can be positive semidefinite if and only if $\bs{\Sigma}_{\bs{q}\bs{p}}(u_t, \bs{s}_t) = \bs{0}$. Since we assumed $\expect*{t-1}{\bs{\Sigma}_{\bs{q}\bs{q}}(u_t, \bs{s}_t)}, \expect*{t-1}{\bs{\Sigma}_{\bs{p}\bs{p}}(u_t, \bs{s}_t)} \neq \bs{0}$, this cannot hold, rendering the assumption empty for $\rho = 0$. Indeed, the regret bound in \Cref{thm: regret upper bound} improves as $\rho$ nears zero, due to the assumption injecting stronger disambiguation between popularity and quality.

Finally, we show an example for which \Cref{assumption: lambda min} does not hold. In such cases disambiguation between quality and popularity is not possible, as shown by our regret lower bound in \Cref{theorem: lower bound}. Consider, for example a case where 
\begin{align*}
        \expect*{t-1}{
        \begin{pmatrix}
            \bs{\Sigma}_{\bs{q}\bs{q}}(u_t, \bs{s}_t) & \bs{\Sigma}_{\bs{q}\bs{p}}(u_t, \bs{s}_t) \\ \bs{\Sigma}_{\bs{p}\bs{q}}(u_t, \bs{s}_t) & \bs{\Sigma}_{\bs{p}\bs{p}}(u_t, \bs{s}_t)
        \end{pmatrix}
        }
        =
        \begin{pmatrix}
            \bs{1}_{d_q \times d_q} & \bs{1}_{d_q \times d_p} \\
            \bs{1}_{d_p \times d_q} & \bs{1}_{d_p \times d_p}
        \end{pmatrix}
        =
        \bs{1}_{d \times d}.
\end{align*}
Clearly, $\bs{1}_{d \times d}$ is a rank-one matrix, and is positive-semidefinite. Considering \Cref{assumption: lambda min}, letting $\bs{x} = \brk*{1, 0, 0, \hdots, 0, -1}^T$, we have that
\begin{align*}
    \bs{x}^T
    \begin{pmatrix}
            \rho \bs{1}_{d_q \times d_q} & \bs{1}_{d_q \times d_p} \\
            \bs{1}_{d_p \times d_q} & \bs{1}_{d_p \times d_p}
    \end{pmatrix}
    \bs{x}
    =
    \rho - 1 < 0, \forall \rho \in (0,1).
\end{align*}
Hence, \Cref{assumption: lambda min} does not hold for any $\rho \in (0,1)$.

To conclude, we believe \Cref{assumption: lambda min} to be a reasonable assumption for most use cases. It is less demanding than previous work and captures the hardness of our problem in terms of disambiguation between quality and popularity, as captured by our regret bound in \Cref{thm: regret upper bound}.

\newpage
\section{Ablations}
\label{appendix: experiments}

We tested the QP-Ranker in \Cref{alg: qrp-ucb} on a series of synthetic environments. For all our experiments we uniformly sampled parameters and embeddings in $\brk[s]*{0, \frac{1}{\sqrt{d}}}^d$. We also scaled the popularity bias, $\popularitybias$ w.r.t. $b_{\max}$ to understand the affect of this scaling. We used the same dimension for $d_q$ and $d_p$, and denote both of them as $d$ here. We varied over problem parameters including $M, d, \alpha_{\min}$, and $b_{\max}$. For each experiment we fixed all parameters to default values, as shown in the table below.
\begin{table}[h!]
    \centering
    \begin{tabular}{c|c|c|c|c}
        $M$ & $d$ & $\alpha_{\min}$ & $b_{\max}$ & $T$\\
        \hline
        $3$ & $8$ & $0.02$ & $0.2$ & $10000$
    \end{tabular}
\end{table}

\Cref{fig: regret parameters} depicts regret of the QP-Ranker w.r.t. variations in $M, d, \alpha_{\min}$, and $b_{\max}$. While \Cref{thm: regret upper bound} shows a $M^3$ dependence in the regret, we found that the dependence to be much better in our synthetic environments. Additionally, we found that, while $\alpha_{\min}$ did significantly increased overall regret (due to its effect in the constant in \Cref{theorem: asymptotic optimality}), we found this effect to diminish. We believe this is due to the ranker not needing to explore the full corpus, for which case saturation of a few items is enough to achieve the desired result. Finally, we found $b_{\max}$ to strongly affect overall regret. This is expected, as lower values of $b_{\max}$ result in a lower upper bound on $\disp$ (values lower than $1$).

\begin{figure}[t!]
    \centering
    \begin{subfigure}[b]{0.45\textwidth}
    \centering
    \includegraphics[width=\linewidth]{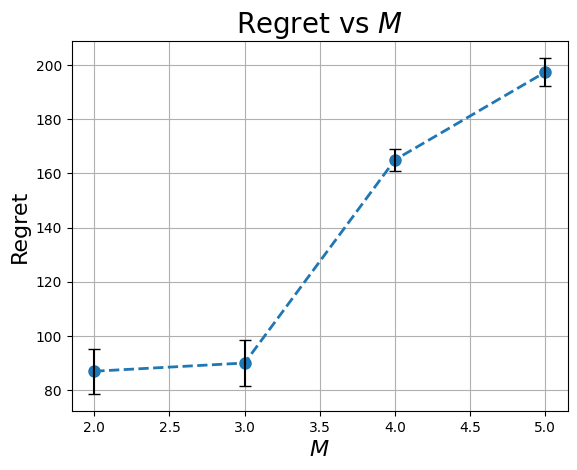}
    \end{subfigure}
    \hspace*{1cm}
    \begin{subfigure}[b]{0.45\textwidth}
    \centering
    \includegraphics[width=\linewidth]{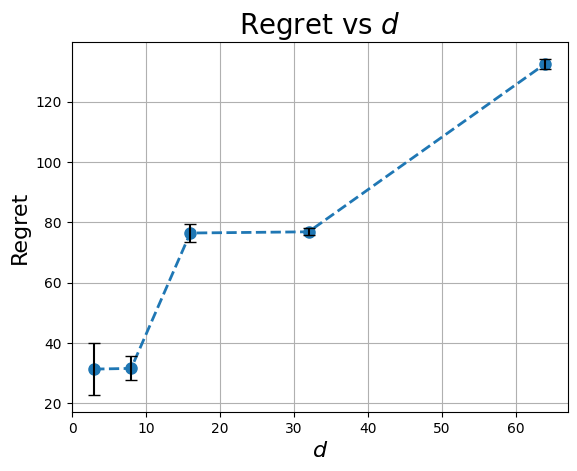}
    \end{subfigure}
    \begin{subfigure}[b]{0.45\textwidth}
    \centering
    \includegraphics[width=\linewidth]{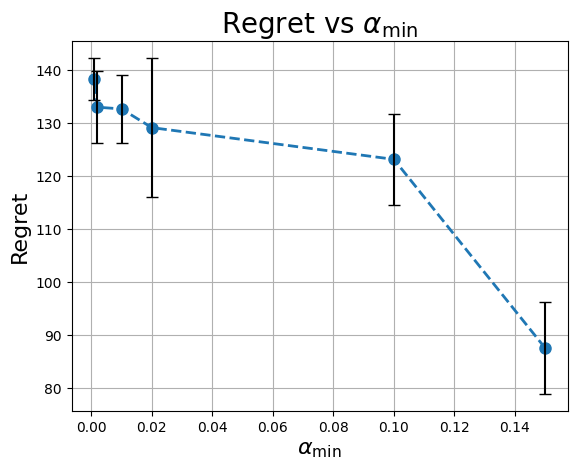}
    \end{subfigure}
    \hspace*{1cm}
    \begin{subfigure}[b]{0.45\textwidth}
    \centering
    \includegraphics[width=\linewidth]{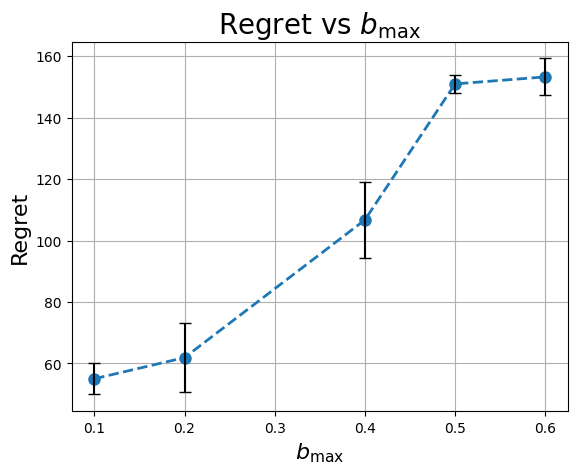}
    \end{subfigure}
    \caption{Varying parameters of the QP-Ranker. Results show average regret over final 100 time steps, averaged over 10 seeds. Our results exhibit behavior consistent with our regret.}
    \label{fig: regret parameters}
\end{figure}

\newpage
\section{Proof of Theorem~\ref{theorem: lower bound}}
\label{appendix: impossibility proof}

\begin{proof}
    Without loss of generality, consider the case of a slate containing a single item ($M=1$) and a corpus of two items $\gD = \brk[c]*{I_1,I_2}$. We assume the problem can be one of two instances; the first instance $\nu_1$ with quality biases $(\qualitybias^1(I_1),\qualitybias^1(I_1))=(\epsilon,-\epsilon)$, for some small $\epsilon\in(0,1]$, and the second instance $\nu_2$ with  $(\qualitybias^1(I_1),\qualitybias^1(I_1))=(-\epsilon,\epsilon)$. At every time an item is played, its popularity bias increases by at least $\alpha_{\min}\in(0,1]$, and the maximal popularity biases are $(b(I_1),b(I_2)) = (1-\epsilon,1+\epsilon)$ in instance $\nu_1$ and $(b(I_1),b(I_2)) = (1+\epsilon,1-\epsilon)$ in instance $\nu_2$. Notice that in both problems, when an item reaches its maximal popularity bias, we have $\qualitybias(I)+\popularitybias(I)=1=\disp_{\max}$. 
    
    We denote by $n_t^p(I)$, the number of times item $I$ was played up to time $t-1$, and by $n_t(I)$, the number of time it was selected --generated a reward of $1$. We also denote the probability that an item $I$ was selected at round $t$ under instance $i$ by $\mu^i_t(I)=\frac{\exp\brk{\qualitybias^i(I)+\popularitybias^i_t(I)}}{1+\exp\brk{\qualitybias^i(I)+\popularitybias^i_t(I)}}$. 
    
    Now, recall that if $x\in[-1,1]$ and $y\in[0,1-x]$, using the monotonicity and Lipchitz constant of the logistic function, one can bound
    \begin{align*}
        \frac{\exp(x+y)}{1+\exp(x+y)} \ge \frac{\exp(x)}{1+\exp(x)}+\frac{e}{(1+e)^2}y\ge \frac{\exp(x)}{1+\exp(x)}
        \ge 0.25 +\frac{y}{6}
    \end{align*}
    
    In particular, given an item $I$ of quality $q\in[-1,1]$, its selection probability lies in the interval
    \begin{align*}
        \mu_t(I)\in\brk[s]*{\frac{\exp\brk{q+\alpha_{\min}n_t(I)}}{1+\exp\brk{q+\alpha_{\min}n_t(I)}},\frac{\exp\brk{\disp_{\max}}}{1+\exp\brk{\disp_{\max}}}}
        \subseteq \brk[s]*{0.25+\frac{\alpha_{\min}}{6}n_t(I),0.75},
    \end{align*}
    where by convention, if the interval is empty then $\mu_t(I)\triangleq\mu^s=0.75$. Similarly, we denote the distribution of the item click rate in problem $i$ at round $t$ by $\nu_{i,t}$
    
    Assume any fixed strategy and let $H_t$ denote the history of the decision process up to time $t$, including all internal randomization up to time $t+1$. Specifically, we denote the selection outcome at time $t$ by $Y_t$ (namely, $Y_t=1$ if the item was selected) and by $U_t$, the internal randomization for time $t+1$ . Finally, let $\sP_\nu$ denote the probability measure w.r.t. arm distribution $\nu$. Following the derivation of inequality (6) in \citep{garivier2019explore}, we have by the chain rule of KL divergence that
    \begin{align*}
        KL\brk{\sP_{\nu_1}^{H_{t+1}},\sP_{\nu_2}^{H_{t+1}}}
        &=
        KL\brk{\sP_{\nu_1}^{(H_t,Y_{t+1},U_{t+1})},\sP_{\nu_2}^{(H_t,Y_{t+1},U_{t+1})}}\\
        &= 
        KL\brk{\sP_{\nu_1}^{H_t},\sP_{\nu_2}^{H_t}} + KL\brk{\sP_{\nu_1}^{(Y_{t+1},U_{t+1})\vert H_t},\sP_{\nu_2}^{(Y_{t+1},U_{t+1})\vert H_t}}
    \end{align*}

Next, we can write
\begin{align*}
    KL\brk{\sP_{\nu_1}^{(Y_{t+1},U_{t+1})\vert H_t},\sP_{\nu_2}^{(Y_{t+1},U_{t+1})\vert H_t}}
    &= \expect*{\nu_1}{\expect*{\nu_1}{ KL\brk{\nu_{1,t}(I_t)\otimes U_{t+1},\nu_{2,t}(I_t)\otimes U_{t+1}\vert H_t}}} \\
    &= \expect*{\nu_1}{\expect*{\nu_1}{ KL\brk{\nu_{1,t}(I_t),\nu_{2,t}(I_t)\vert H_t}}} \\
    &= \expect*{\nu_1}{\indicator{I_t=I_1}kl\brk*{\mu_t^1(I_1),\mu_t^2(I_1)}+ \indicator{I_t=I_2}kl\brk*{\mu_t^1(I_2),\mu_t^2(I_2)}},
\end{align*}
where $kl(\mu_1,\mu_2)$ is the KL divergence between two Bernoulli random variables of means $\mu_1,\mu_2$. Importantly, notice that the expectation follows the history given problem instance $\nu_1$, thus, if under this instance $n_t(I)>\frac{3}{\alpha_{\min}}$, we will have $\mu_t^1(I)=\mu_t^2(I)=\mu^s$ and $kl(\mu_t^1(I),\mu_t^2(I))=0$. Otherwise, since arm means are bounded in $[0.25,0.75]$, one can easily see that 
\begin{align*}
    kl(\mu_t^1(I),\mu_t^2(I)) 
    \leq \max\brk[c]{kl(0.25,0.75), kl(0.75,0.25)}
    \leq 1
\end{align*}
Substituting back yields the bound
\begin{align*}
    KL&\brk{\sP_{\nu_1}^{(Y_{t+1},U_{t+1})\vert H_t},\sP_{\nu_2}^{(Y_{t+1},U_{t+1})\vert H_t}} \\
    & \leq \expect*{\nu_1}{\indicator{I_t=I_1,n_t(I_1)\leq \frac{3}{\alpha_{\min}}} + \indicator{I_t=I_2,n_t(I_2)\leq \frac{3}{\alpha_{\min}}}},
\end{align*}
and iterating over the chain rule of the KL divergence yields
\begin{align*}
     KL\brk{\sP_{\nu_1}^{H_{t+1}},\sP_{\nu_2}^{H_{t+1}}}
     &\leq \expect*{\nu_1}{\sum_{s=1}^t\indicator{I_t=I_1,n_t(I_1)\leq \frac{3}{\alpha_{\min}}}} 
     + \expect*{\nu_1}{\indicator{I_t=I_2,n_t(I_2)\leq \frac{3}{\alpha_{\min}}}}.
\end{align*}
We move forward and bound each of these sums. In particular, one can write
\begin{align*}
    \expect*{\nu_1}{\indicator{I_t=I,n_t(I)\leq \frac{3}{\alpha_{\min}}}}
    &= \expect*{\nu_1}{\indicator{Y_{t+1}=1,I_t=I,n_t(I)\leq \frac{3}{\alpha_{\min}}}}\\
    &\quad+ \expect*{\nu_1}{\indicator{Y_{t+1}=0,I_t=I,n_t(I)\leq \frac{3}{\alpha_{\min}}}} \\
    &\leq \expect*{\nu_1}{\indicator{Y_{t}=1,I_t=I,n_t(I)\leq \frac{3}{\alpha_{\min}}}}\\
    &\quad+ 0.75\expect*{\nu_1}{\indicator{I_t=I,n_t(I)\leq \frac{3}{\alpha_{\min}}}}.
\end{align*}
Reorganizing and noticing that $Y_{t}=1$ implies that $n_{t+1}^c(I)=n_t(I)+1$, we get 
\begin{align*}
    &\expect*{\nu_1}{\indicator{I_t=I,n_t(I)\leq \frac{3}{\alpha_{\min}}}}
    \leq 4 \expect*{\nu_1}{\indicator{n_{t+1}^c(I)=n_t(I)+1,n_t(I)\leq \frac{3}{\alpha_{\min}}}},
\end{align*}
and summing while noting that both events in the indicator can be active only for $ \frac{3}{\alpha_{\min}}+1$ times yields
\begin{align*}
     KL\brk{\sP_{\nu_1}^{H_{t+1}},\sP_{\nu_2}^{H_{t+1}}}
     &\leq 8\brk*{\frac{3}{\alpha_{\min}}+1} \leq \frac{25}{\alpha_{\min}}.
\end{align*}
The next step of the proof is to apply the data-processing inequality for the KL divergence (Lemma 1 of \citealt{garivier2019explore}), with the random variables $Z=\frac{n_{t+1}(I)}{t}\in[0,1]$ for any $I\in\brk[c]{I_1,I_2}$. By doing so, one get that
\begin{align}
    kl\brk*{\frac{\expect{\nu_1}{n_{t+1}^p(I)}}{t},\frac{\expect{\nu_2}{n_{t+1}^p(I)}}{t}} \leq KL\brk{\sP_{\nu_1}^{H_{t+1}},\sP_{\nu_2}^{H_{t+1}}}
     &\leq \frac{25}{\alpha_{\min}}.
\end{align}
Finally, recall that for any $p,q\in(0,1)$, we have
\begin{align*}
    kl(p,q) 
    &= p\ln\frac{p}{q}+(1-p)\ln\frac{1-p}{1-q} \\
    &= p\ln\frac{1}{q} +\underbrace{\brk{p\ln p+(1-p)\ln(1-p)}}_{\ge-\ln2}+ \underbrace{(1-p)\ln\frac{1}{1-q}}_{\ge0} \\
    & \ge p\ln\frac{1}{q}-\ln 2.
\end{align*}
Substituting back, we have 
\begin{align*}
    \expect{\nu_1}{n_{t+1}^p(I_1)}
    \leq \brk*{\frac{25}{\alpha_{\min}} + \ln2}\brk*{\ln\frac{t}{\expect{\nu_2}{n_{t+1}^p(I_1)}}}^{-1}\cdot t
    \leq \frac{26}{\alpha_{\min}}\brk*{\ln\frac{t}{\expect{\nu_2}{n_{t+1}^p(I_1)}}}^{-1}\cdot t.
\end{align*}
Specifically, either that $\expect{\nu_2}{n_{t+1}^p(I_1)}\ge \exp\brk{-52/\alpha_{\min}}t = \Omega(t)$ or, through direct substitution, $\expect{\nu_1}{n_{t+1}^p(I_1)}\le t/2$, i.e., $\expect{\nu_1}{n_{t+1}^p(I_2)} = \Omega(t)$. In other words, any algorithm would choose the suboptimal item a linear number of times in at least one of the two problems, and since the items have strictly positive quality gap, this would incur a linear regret.
\end{proof}

\newpage
\section{Proof of Theorem~\ref{theorem: asymptotic optimality}}

\begin{proof}

We denote $p_{\min} = \min_{i \in [M], u_t \in \gU \bs{s}_t \in \gS, t \geq 1}z_i(\qualitybias\brk*{u_t, \bs{s}_t} + \rankbias + \popularitybias\brk*{u_t, \bs{s}_t})$. Notice that due to boundness of $\disp_t$ for all $t$, $p_{\text{min}} \geq \Omega\brk*{\frac{1}{M}}$. 

Next, we let $N_t(I)$ denote the number of times item $I$ was selected by the \emph{ranker} up to time $t$. That is, $N_t(I) = \sum_{k=1}^t \indicator{I \in \bs{s}_k}$. Note that it is not necessarily true that $N_t(I) = n_t(I)$. 

Let $I \in \gD, m \in [T]$, and let $\tau^*_{i}(I) = \min\brk[c]*{t \geq 1: N_t(I) = i}$. For any $m \in [T]$, we define the coupling $\brk[c]*{(X_i, Y_i)}_{i=1}^m$, where $\brk[c]*{X_i}_{i=1}^m$ are iid Bernouli random variables with probability $p_{\min}$ and $Y_i = \indicator{c_{\tau^*_{i}(I)} = I}$. Then, for any $a \in \R$, 
\begin{align}
\label{eq: coupling}
    P\brk*{\sum_{i=1}^m Y_i \leq a, N_t(I) = m}
    \leq
    P\brk*{\sum_{i=1}^m X_i \leq a, N_t(I) = m}.
\end{align}

Denote $m_0 = \frac{8b_{\max}}{p_{\min}\alpha_{\min}}\log(\frac{3\abs{\gD}}{\delta})$ and let $t \geq m_0$. Then,
\begin{align*}
&P\brk*{n_t(I) \leq \frac{b_{\max}}{\alpha_{\min}}, N_t(I)\geq m_0} \\
&=
P\brk*{\sum_{k=1}^t \indicator{c_k = I} \leq \frac{b_{\max}}{\alpha_{\min}}, N_t(I)\geq m_0} \\
&\leq
\sum_{m=m_0}^\infty
P\brk*{\sum_{k=1}^t \indicator{c_k = I} \leq \frac{b_{\max}}{\alpha_{\min}}, N_t(I) = m} 
\tag{\text{Union Bound}} \\
&\leq
\sum_{m=m_0}^\infty
P\brk*{\sum_{k : I \in S_k} \indicator{c_k = I} \leq \frac{b_{\max}}{\alpha_{\min}}, N_t(I) = m} \\
&=
\sum_{m=m_0}^\infty
P\brk*{\sum_{i=1}^{N_t(I)} \indicator{c_{\tau^*_{i}(I)} = I} \leq \frac{b_{\max}}{\alpha_{\min}}, N_t(I) = m} \\
&\leq
\sum_{m=m_0}^\infty
P\brk*{\sum_{i=1}^{N_t(I)} X_i \leq \frac{b_{\max}}{\alpha_{\min}}, N_t(I) = m} \tag{\Cref{eq: coupling}} \\
&=
\sum_{m=m_0}^\infty
P\brk*{\sum_{i=1}^{m} X_i \leq \frac{b_{\max}}{\alpha_{\min}}, N_t(I) = m} \\
&\leq
\sum_{m=m_0}^\infty
P\brk*{\sum_{i=1}^{m} X_i \leq \frac{b_{\max}}{\alpha_{\min}}}
\end{align*}
By Hoeffdings inequality, since $X_i \overset{\mathrm{iid}}{\sim} \text{Bern}\brk*{p_{\min}}$,
\begin{align*}
    P\brk*{\sum_{i=1}^{m} X_i - p_{\min}m \leq -\frac{p_{\min}m}{2}} \leq \exp\brk[c]*{-m/2}.
\end{align*}
Therefore,
\begin{align*}
    P\brk*{\sum_{i=1}^{m} X_i \leq \frac{p_{\min}m}{2}} \leq \exp\brk[c]*{-m/2}.
\end{align*}
Using the above and the definition of $m_0$ we get that
\begin{align*}
    P\brk*{n_t(I) \leq \frac{b_{\max}}{\alpha_{\min}}, N_t(I)\geq m_0}
    &\leq
    \sum_{m=m_0}^\infty
    P\brk*{\sum_{i=1}^{m} X_i \leq \frac{b_{\max}}{\alpha_{\min}}} \\
    &\leq
    \sum_{m=m_0}^\infty
    P\brk*{\sum_{i=1}^{m} X_i \leq \frac{p_{\min}m}{2}} \\
    &\leq
    \sum_{m=m_0}^\infty
    \exp\brk[c]*{-m/2} \\
    &\leq
    3\exp\brk[c]*{-m_0/2} 
    \leq 
    \frac{\delta}{\abs{\gD}}.
\end{align*}
By the union bound, for any $t \geq m_0$
\begin{align*}
    P\brk*{\bigcup_{I \in \gD} \brk[c]*{n_t(I) \leq \frac{b_{\max}}{\alpha_{\min}}, N_t(I)\geq m_0}} \leq \delta.
\end{align*}

By \Cref{assump: popularity bias}, $\brk[c]*{\brk[s]*{\popularitybias_{t}\brk*{u, \bs{s}_t}}_i < b_i(u, \bs{s}_t)} \subseteq \brk[c]*{n_t(I_{i,t}) \leq \frac{b_{\max}}{\alpha_{\min}}}$. Therefore, it also holds that
\begin{align*}
    P\brk*{\bigcup_{I \in \gD}
    \brk[c]*{\exists i\le M: I_{i,t}=I, \brk[s]*{\popularitybias_{t}\brk*{u, \bs{s}_t}}_i < b_i(u, \bs{s}_t), N_t(I)\geq m_0}} \leq \delta.
\end{align*} 
In other words, w.p. at least $1-\delta$, for all $I \in \gD$ and $i\le M$, if item $I$ is in the slate ($I_{i,t}=I$) and $N_t(I)\geq m_0$, then its popularity bias is saturated ($\brk[s]*{\popularitybias_{t}\brk*{u, \bs{s}_t}}_i = b_i(u, \bs{s}_t)$). On the other hand, if an item is in the slate and $N_t(I)< m_0$, then $N_t(I)$ increases by $1$. Thus, by the pigeonhole principle, the number of rounds such that an item in the slate has $N_t(I)< m_0$ is bounded by $(m_0+1)\abs{\gD}$.

Finally, noticing that when all items in the slate are saturated, the quality ranker is optimal, and otherwise, the instantaneous regret is bounded by $1$, we get w.p. at least $1-\delta$ that
\begin{align*}
    \text{Reg}_{\gR_q}(T)
    &\leq
    \sum_{t=1}^T\indicator{\exists i\le M: \brk[s]*{\popularitybias_{t}\brk*{u, \bs{s}_t}}_i < b_i(u, \bs{s}_t)}\\
    &\leq
    \sum_{t=1}^T\indicator{\exists i\le M: I_{i,t}=I, N_t(I)< m_0}\\
    &\leq 
    \min\brk[c]*{(m_0+1)\abs{\gD}, T}.
\end{align*}
Plugging in the definition of $m_0$ and using the fact that $p_{\text{min}} \geq \Omega\brk*{\frac{1}{M}}$ completes the proof.

\end{proof}

\newpage
\section{Proof of Theorem~\ref{thm: regret upper bound}}
\label{appendix: regret proof}

For clarity, we denote 
\begin{align*}
    v_t(u, \bs{s}, \bs{\theta}, \bs{\phi}) 
    = 
    \sum_{i=1}^M 
    \brk[s]*{\qualitybias_{\bs{\theta}}(u, \bs{s})}_i
    z_i\brk*{\qualitybias_{\bs{\theta}}(u, \bs{s}) + \popularitybias_{t, \bs{\phi}}(u, \bs{s}, \bs{c}_{1:t-1}) + \rankbias}.
\end{align*}
and
\begin{align*}
    v_{\text{sat}}(u, \bs{s}, \bs{\theta}, \bs{\phi}) 
    = 
    \sum_{i=1}^M 
    \brk[s]*{\qualitybias_{\bs{\theta}}(u, \bs{s})}_i
    z_i\brk*{\qualitybias_{\bs{\theta}}(u, \bs{s}) + \bs{b}_{\bs{\phi}}(u, \bs{s}) + \rankbias}.
\end{align*}
Also, recall that 
\begin{align}
\label{eq: UCB policy}
    \bs{s}_t
    \in 
    \arg\max_{\bs{s} \in \gS} 
    v_{\text{sat}}(u_t, \bs{s}, \bs{\theta}_t, \bs{\phi}_t) 
    + 
    \epsilon_t(\bs{s}).
\end{align}
Let $\gE^{\text{sat}}_t = \brk[c]*{ \brk[c]*{\popularitybias_{t}\brk*{u, \bs{s}_t, \bs{c}_{1:t-1}} = \bs{b}(u, \bs{s}_t)}}.$ Similar to \Cref{lemma: error bound}, with probability at least $1-\delta$, we have that
\begin{align*}
    \sum_{t=1}^T \indicator{\brk*{\gE^{\text{sat}}_t}^c}
    \leq
    \frac{8\abs{\gD}Mb_{\max}}{\alpha_{\min}}\log\brk*{\frac{2\abs{\gD}}{\delta}}.
\end{align*}
Then, with probability at least $1-\delta$,
\begin{align*}
    \text{Reg}_{\gR_{qp}}(T)
    &=
    \sum_{t=1}^T 
    \expect*{u_t}
    {
        \brk*{
            v_t(u_t, \bs{s}^q, \bs{\theta}^*, \bs{\phi}^*)
            -
            v_t(u_t, \bs{s}_t, \bs{\theta}^*, \bs{\phi}^*)
        }
    } \\
    &\leq
    \sum_{t=1}^T 
    \expect*{u_t}
    {
        \brk*{
            v_t(u_t, \bs{s}^q, \bs{\theta}^*, \bs{\phi}^*)
            -
            v_t(u_t, \bs{s}_t, \bs{\theta}^*, \bs{\phi}^*)
        }
        \indicator{\gE^{\text{sat}}_t}
    }
    +
    \sum_{t=1}^T 2 \indicator{\brk*{\gE^{\text{sat}}_t}^c}
    \\
    &\leq
    \sum_{t=1}^T 
    \expect*{u_t}
    {
        v_{\text{sat}}(u_t, \bs{s}^q, \bs{\theta}^*, \bs{\phi}^*)
        -
        v_{\text{sat}}(u_t, \bs{s}_t, \bs{\theta}^*, \bs{\phi}^*)
    }
    +
    \frac{16\abs{\gD}Mb_{\max}}{\alpha_{\min}}\log\brk*{\frac{2\abs{\gD}}{\delta}} \\
    &\leq
    \sum_{t=1}^T 
    \expect*{u_t}
    {
        v_{\text{sat}}(u_t, \bs{s}^q, \bs{\theta}_t, \bs{\phi}_t)
        +
        \epsilon_t\brk*{\bs{s}^q}
        -
        v_{\text{sat}}(u_t, \bs{s}_t, \bs{\theta}_t, \bs{\phi}_t)
        +
        \epsilon_t\brk*{\bs{s}_t}
    }
    +
    \frac{16\abs{\gD}Mb_{\max}}{\alpha_{\min}}\log\brk*{\frac{2\abs{\gD}}{\delta}}
    \tag{\Cref{lemma: error bound}} \\
    & \leq
    \sum_{t=1}^T 
    \expect*{u_t}
    {
        v_{\text{sat}}(u_t, \bs{s}_t, \bs{\theta}_t, \bs{\phi}_t)
        +
        \epsilon_t\brk*{\bs{s}_t}
        -
        v_{\text{sat}}(u_t, \bs{s}_t, \bs{\theta}_t, \bs{\phi}_t)
        +
        \epsilon_t\brk*{\bs{s}_t}
    } 
    +
    \frac{16\abs{\gD}Mb_{\max}}{\alpha_{\min}}\log\brk*{\frac{2\abs{\gD}}{\delta}}
    \tag{\Cref{eq: UCB policy}} \\
    &=
    2 
    \sum_{t=1}^T 
    \epsilon_t\brk*{\bs{s}_t} 
    +
    \frac{16\abs{\gD}Mb_{\max}}{\alpha_{\min}}\log\brk*{\frac{2\abs{\gD}}{\delta}} \\
    &\leq
    2 
    \sum_{t=1}^T 
    \brk*{4\sqrt{M} + \frac{1}{\sqrt{1-\rho_{min}}}}
    \gamma_t(\delta)
    \sqrt{
        \expect*{t-1}{
            \norm{
                \bs{x}(u, \bs{s})
            }_{\bs{V}_t^{-1}}^2
    }}
    +
    \frac{16\abs{\gD}Mb_{\max}}{\alpha_{\min}}\log\brk*{\frac{2\abs{\gD}}{\delta}} \\
    &\leq
    2 \brk*{4\sqrt{M} + \frac{1}{\sqrt{1-\rho_{min}}}}
    \gamma_T(\delta)
    \sqrt{
        T 
        \expect*{t-1}{
        \sum_{t=1}^T
            \norm{
                \bs{x}(u, \bs{s})
            }_{\bs{V}_t^{-1}}^2
    }}
    +
    \frac{16\abs{\gD}Mb_{\max}}{\alpha_{\min}}\log\brk*{\frac{2\abs{\gD}}{\delta}}
    \tag{Cauchy–Schwarz ineq.} \\
    &\leq
    \gO\brk*{
    \brk*{\sqrt{M} + \frac{1}{\sqrt{1-\rho_{min}}}}
    \gamma_T(\delta)
    \sqrt{
        (d_q + d_p)T\log\brk*{1 + \frac{T}{\lambda (d_q + d_p)}
    }}
    }
\end{align*}
where the last inequality is due to \citet{abbasi2011improved} (Lemma 11). Finally letting $\lambda \leq \gO\brk*{\frac{1}{\sqrt{L}}}$ yields the desired result.
\begin{align*}
    \text{Reg}_{\gR_{qp}}(T)
    &\leq
    \gO\brk*{M^2
    \brk*{\sqrt{M} + \frac{1}{\sqrt{1-\rho_{min}}}}
    \sqrt{\log(1/\delta) + Md\log\brk*{1+\frac{TL}{d}}}
    \sqrt{
        d T\log\brk*{1 + \frac{TL}{d}
    }}
    } \\
    &\leq
    \gO\brk*{M^{2.5}
    d\sqrt{T}
    \brk*{\sqrt{M} + \frac{1}{\sqrt{1-\rho_{min}}}}
    \sqrt{\log(1/\delta)}
     \log\brk*{1 + \frac{TL}{d}
    }}
\end{align*}

\begin{lemma}
\label{lemma: error bound}
    For any $\bs{s} \in \gS, t \in [T], \delta \in (0,1)$, with probability at least $1-\delta$
    \begin{align*}
        \abs{
        \expect*{u}
        {
            v_{\text{sat}}(u, \bs{s}, \bs{\theta}^*, \bs{\phi}^*)
            -
            v_{\text{sat}}(u, \bs{s}, \bs{\theta}_t, \bs{\phi}_t)
        }}
        \leq
        \epsilon_t(\bs{s})
    \end{align*}
\end{lemma}
\begin{proof}
    We have that
    \begin{align*}
        &\abs{
        \expect*{u}
        {
            v_{\text{sat}}(u, \bs{s}, \bs{\theta}^*, \bs{\phi}^*)
            -
            v_{\text{sat}}(u, \bs{s}, \bs{\theta}_t, \bs{\phi}_t)
        }}\\
        &=
        \abs{
        \expect*{u}
        {
            \sum_{i=1}^M 
            \brk[s]*{\qualitybias_{\bs{\theta}^*}(u, \bs{s})}_i
            z_i\brk*{\qualitybias_{\bs{\theta}^*}(u, \bs{s}) + \bs{b}_{\bs{\phi}^*}(u, \bs{s}) + \rankbias}
            -
            \brk[s]*{\qualitybias_{\bs{\theta}_t}(u, \bs{s})}_i
            z_i\brk*{\qualitybias_{\bs{\theta}_t}(u, \bs{s}) + \bs{b}_{\bs{\phi}_t}(u, \bs{s}) + \rankbias}
            }} \\
        &\leq
        \abs{
        \expect*{u}
        {
            \sum_{i=1}^M 
            \brk[s]*{\qualitybias_{\bs{\theta}^*}(u, \bs{s})}_i
            \brk3{
            z_i\brk*{\qualitybias_{\bs{\theta}^*}(u, \bs{s}) + \bs{b}_{\bs{\phi}^*}(u, \bs{s}) + \rankbias}
            -
            z_i\brk*{\qualitybias_{\bs{\theta}_t}(u, \bs{s}) + \bs{b}_{\bs{\phi}_t}(u, \bs{s}) + \rankbias}}}} \\
        &~~~+
        \abs{
        \expect*{u}
        {
            \sum_{i=1}^M 
            \brk3{
            \brk[s]*{\qualitybias_{\bs{\theta}^*}(u, \bs{s})}_i
            -
            \brk[s]*{\qualitybias_{\bs{\theta}_t}(u, \bs{s})}_i
            }
            z_i\brk*{\qualitybias_{\bs{\theta}_t}(u, \bs{s}) + \bs{b}_{\bs{\phi}_t}(u, \bs{s}) + \rankbias}}}
        \tag{Triangle Ineq.}
    \end{align*}
    By \Cref{lemma: multinomial bandit upper bound}, 
    \begin{align*}
        &\abs{
        \expect*{u}
        {
            \sum_{i=1}^M 
            \brk[s]*{\qualitybias_{\bs{\theta}^*}(u, \bs{s})}_i
            \brk3{
            z_i\brk*{\qualitybias_{\bs{\theta}^*}(u, \bs{s}) + \bs{b}_{\bs{\phi}^*}(u, \bs{s}) + \rankbias}
            -
            z_i\brk*{\qualitybias_{\bs{\theta}_t}(u, \bs{s}) + \bs{b}_{\bs{\phi}_t}(u, \bs{s}) + \rankbias}
        }}} \\
        &\leq
        4\sqrt{M}\gamma_t(\delta)
        \expect*{u}{
            \norm{
                \bs{x}(u, \bs{s})
            }_{\bs{V}_t^{-1}}
        } \\
        &\leq
        4\sqrt{M}\gamma_t(\delta)
        \sqrt{
        \expect*{u}{
            \norm{
                \bs{x}(u, \bs{s})
            }_{\bs{V}_t^{-1}}^2
        }}.
    \end{align*}
    By Hölder's inequality and \Cref{lemma: quality bound},
    \begin{align*}
        &\abs{
        \expect*{u}
        {
            \sum_{i=1}^M 
            \brk3{
            \brk[s]*{\qualitybias_{\bs{\theta}^*}(u, \bs{s})}_i
            -
            \brk[s]*{\qualitybias_{\bs{\theta}_t}(u, \bs{s})}_i
            }
            z_i\brk*{\qualitybias_{\bs{\theta}_t}(u, \bs{s}) + \bs{b}_{\bs{\phi}_t}(u, \bs{s}) + \rankbias}}
        } \\
        &\leq
        \norm{
        \expect*{u}
        {
            \qualitybias_{\bs{\theta}^*}(u, \bs{s})
            -
            \qualitybias_{\bs{\theta}_t}(u, \bs{s})
        }}_\infty
        \norm{
        \expect*{u}
        {
            \bs{z}\brk*{\qualitybias_{\bs{\theta}_t}(u, \bs{s}) + \bs{b}_{\bs{\phi}_t}(u, \bs{s}) + \rankbias}
        }}_1 \\
        &=
        \norm{
        \expect*{u}
        {
            \qualitybias_{\bs{\theta}^*}(u, \bs{s})
            -
            \qualitybias_{\bs{\theta}_t}(u, \bs{s})
        }}_\infty \\
        &\leq
        \frac{\gamma_t(\delta)}{\sqrt{1-\rho_{min}}}
        \sqrt{
        \expect*{u}{
            \norm{
                \bs{x}(u, \bs{s})
            }_{\bs{V}_t^{-1}}^2
        }}.
     \end{align*}
    Combining the above we get the desired result.
\end{proof}

\begin{lemma}
\label{lemma: quality bound}
    For $t > 1$
    \begin{align*}
        \abs{
        \expect*{t-1}
        {
            \brk[s]*{\qualitybias_{\bs{\theta}^*}(u, \bs{s})}_i
            -
            \brk[s]*{\qualitybias_{\bs{\theta}_t}(u, \bs{s})}_i
        }}
        \leq
        \frac{\gamma_t(\delta)}{\sqrt{1-\rho_{min}}}
        \sqrt{
        \expect*{t-1}{
            \norm{
                \bs{x}(u, \bs{s})
            }_{\bs{V}_t^{-1}}^2
        }} 
    \end{align*}
\end{lemma}
\begin{proof}
    For any $i \in [M], \bs{s} \in \gS, t \in [T]$,
    \begin{align*}
        \abs{
        \expect*{t-1}
        {
            \brk[s]*{\qualitybias_{\bs{\theta}^*}(u, \bs{s})}_i
            -
            \brk[s]*{\qualitybias_{\bs{\theta}_t}(u, \bs{s})}_i
        }}
        &=
        \abs{
        \expect*{t-1}
        {
            \bs{x}_q^T(u, \bs{s}) \bs{\theta}^*_i
            -
            \bs{x}_q^T(u, \bs{s}) \bs{\theta}_{i,t}
        }} \\
        &=
        \abs{
        \expect*{t-1}
        {
            \bs{x}^T(u, \bs{s})
            \bs{D}_{d_q, d_p}
            \brk*{
                \bs{\psi}^*_i
                -
                \bs{\psi}_{i,t}
            }
        }} \\
        &\leq
        \expect*{t-1}{
            \abs{
                \bs{x}^T(u, \bs{s})
                \bs{D}_{d_q, d_p}
                \brk*{
                    \bs{\psi}^*_i
                    -
                    \bs{\psi}_{i,t}
                }
            }
        }\\
        &\leq
        \expect*{t-1}{
            \norm{
                \bs{x}(u, \bs{s})
            }_{\bs{V}_t^{-1}}
            \norm{
                \bs{D}_{d_q, d_p}
                \brk*{
                    \bs{\psi}^*_i
                    -
                    \bs{\psi}_{i,t}
                }
            }_{\bs{V}_{t}}
        }
    \end{align*}
    By \Cref{assumption: lambda min}, we have that
    \begin{align}
        \expect*{t-1}
        {
            \bs{D}_{d_q, d_p}
            \bs{V}_t
            \bs{D}_{d_q, d_p}
        }
        &=
        \bs{D}_{d_q, d_p}
        \brk*{
            \lambda \bs{I}_{d_p + d_q}
            +
            \expect*{t-1}
            {
                \sum_{k=1}^{t-1}
                \bs{x}(u_k, \bs{s}_k) \bs{x}^T(u_k, \bs{s}_k)
            }
        } 
        \bs{D}_{d_q, d_p} \nonumber \\
        &=
        \lambda \bs{D}_{d_q, d_p}
        + 
        \expect*{t-1}{
            \begin{pmatrix}
                \bs{\Sigma}_{\bs{q}\bs{q}}(u, \bs{s}) & \bs{0} \\
                \bs{0} & \bs{0}
            \end{pmatrix}
        } \nonumber \\
        &\preceq
        \lambda \bs{I}_{d_q + d_p}
        + 
        \frac{1}{1-\rho}
        \expect*{t-1}{
            \begin{pmatrix}
                \bs{\Sigma}_{\bs{q}\bs{q}}(u, \bs{s}) & \bs{\Sigma}_{\bs{q}\bs{p}}(u, \bs{s}) \\
                \bs{\Sigma}_{\bs{p}\bs{q}}(u, \bs{s}) & \bs{\Sigma}_{\bs{p}\bs{p}}(u, \bs{s})
            \end{pmatrix}
        } \nonumber \\
        &\preceq
        \frac{1}{1-\rho}
        \bs{V}_t.
        \label{eq: application of lambda min assumption}
    \end{align}
    Then,
    \begin{align*}
        &\expect*{t-1}{
            \norm{
                \bs{x}(u, \bs{s})
            }_{\bs{V}_t^{-1}}
            \norm{
                \bs{D}_{d_q, d_p}
                \brk*{
                    \bs{\psi}^*_i
                    -
                    \bs{\psi}_{i,t}
                }
            }_{\bs{V}_{t}}
        } \\
        &\leq
        \sqrt{
        \expect*{t-1}{
            \norm{
                \bs{x}(u, \bs{s})
            }_{\bs{V}_t^{-1}}^2
        }
        \expect*{t-1}{
            \norm{
                \bs{D}_{d_q, d_p}
                \brk*{
                    \bs{\psi}^*_i
                    -
                    \bs{\psi}_{i,t}
                }
            }_{\bs{V}_{t}}^2
        }} \tag{Cauchy–Schwarz ineq.}\\
        &\leq
        \frac{1}{\sqrt{1-\rho_{min}}}
        \sqrt{
        \expect*{t-1}{
            \norm{
                \bs{x}(u, \bs{s})
            }_{\bs{V}_t^{-1}}^2
        }
        \expect*{t-1}{
            \brk*{
                \norm{
                    \bs{\psi}^*_i
                    -
                    \bs{\psi}_{i,t}
                }_{\bs{V}_{t}}^2
            }
        }}
        \tag{\Cref{eq: application of lambda min assumption}} \\
        &\leq
        \frac{1}{\sqrt{1-\rho_{min}}}
        \sqrt{
        \expect*{t-1}{
            \norm{
                \bs{x}(u, \bs{s})
            }_{\bs{V}_t^{-1}}^2
        }
        \expect*{t-1}{
            \norm{
                \bs{\psi}^*_i
                -
                \bs{\psi}_{t}
            }_{\bs{I}_M \otimes \bs{V}_{t}}^2
        }} \\
        &\leq
        \frac{\gamma_t(\delta)}{\sqrt{1-\rho_{min}}}
        \sqrt{
        \expect*{t-1}{
            \norm{
                \bs{x}(u, \bs{s})
            }_{\bs{V}_t^{-1}}^2
        }}.
        \tag{\Cref{lemma: param design matrix upper bound}} \\
    \end{align*}
\end{proof}
This completes the proof.

\newpage

\section{Useful Lemmas}

\begin{lemma}
\label{lemma: multinomial bandit upper bound}
With probability at least $1-\delta$ 
    \begin{align*}
        \abs{
        \expect*{u}
        {
            \sum_{i=1}^M 
            \brk[s]*{\qualitybias_{\bs{\theta}^*}(u, \bs{s})}_i
            \brk3{
            z_i\brk*{\qualitybias_{\bs{\theta}^*}(u, \bs{s}) + \bs{b}_{\bs{\phi}^*}(u, \bs{s}) + \rankbias}
            -
            z_i\brk*{\qualitybias_{\bs{\theta}_t}(u, \bs{s}) + \bs{b}_{\bs{\phi}_t}(u, \bs{s}) + \rankbias}
        }}} \\
        \leq
        \sqrt{M}\gamma_t(\delta)
        \expect*{u}{
            \norm{
                \bs{x}(u, \bs{s})
            }_{\bs{V}_t^{-1}}
        }.
    \end{align*}
\end{lemma}
\begin{proof}
    The proof follows the same step as the proof in Appendix C.2 of \citet{amani2021ucb}, using the fact that $\norm{\qualitybias}_2 = \sqrt{\sum_{i=1}^M\norm{\qualitybias_i}_2^2} = \sqrt{M}$, and, e.g., Eq. (29) of \citet{amani2021ucb} for which
    $
        \sup \lambda_{\max}(\bs{A})
        \cdot
        \inf 1/\lambda_{\min}(\bs{A})
        \leq
         e^2\brk*{1 + Me}^2
         \leq 
         2e^4M^2.
    $
    Note that this paper uses a different definition of $\gamma_t$ than in \citet{amani2021ucb}.
\end{proof}

\begin{lemma}
\label{lemma: param design matrix upper bound}
    With probability at least $1-\delta$
    \begin{align*}
        \expect*{u}{
            \norm{
                \bs{\psi}^*
                -
                \bs{\psi}_t
            }_{\bs{I}_M \otimes \bs{V}_{t}}^2
        }
        \leq
        \gamma_t(\delta)
    \end{align*}
\end{lemma}
\begin{proof}
    Follows from Lemma 14 of \citet{amani2021ucb}.
\end{proof}

\end{document}